\pdfoutput=1

\documentclass[11pt]{amsart}
\usepackage{appendix}
\usepackage{pgfplots}
\usepackage[verbose=true]{microtype}
\usepackage{amsmath}
\usepackage{amssymb}
\usepackage{amsfonts}
\usepackage{amsthm}
 \usepackage[foot]{amsaddr}
\usepackage{bbm}
\usepackage{mathtools}
\usepackage{color}
\usepackage{xcolor}
\usepackage[colorlinks,linkcolor=black,citecolor=black,linktocpage,hypertexnames=true,pdfpagelabels]{hyperref}
\usepackage[capitalise]{cleveref}
\usepackage[all]{xy}
\usepackage{subcaption}
\usepackage{array}
\usepackage{enumerate}
\usepackage{enumitem}
\usepackage{tikz}
\usepackage{relsize}
\usepackage{diagbox}
\usepackage{tikz-cd} 
\usepackage{braket}
\usepackage{hyperref}
\usepackage{caption}
\usepackage{subcaption}
\usepackage{mdframed}
\usepackage[ruled,vlined]{algorithm2e}

\usepackage{dsfont}
\usepackage{eufrak}

\usetikzlibrary{arrows.meta}
\usetikzlibrary{positioning}
\usetikzlibrary{arrows, decorations.markings}
\tikzset{
vecArrow/.style={
  thick,
  decoration={markings,mark=at position
   1 with {\arrow[scale=2,thin]{open triangle 60}}},
  double distance=1.4pt, shorten >= 10.5pt,
  preaction = {decorate},
  postaction = {draw,line width=1.4pt, white,shorten >= 4.5pt}
  },
innerWhite/.style={
  semithick, 
  white,
  line width=1.4pt, 
  shorten >= 4.5pt
  }
}
\usetikzlibrary{fit,calc}
\tikzstyle{tensor}=[rectangle,draw=blue!50,fill=blue!20,thick]

\newtheorem{theorem}{Theorem}
\newtheorem{lemma}[theorem]{Lemma}
\newtheorem{corollary}[theorem]{Corollary}
\newtheorem{proposition}[theorem]{Proposition}
\newtheorem{problem}[theorem]{Problem}
\newtheorem{conjecture}[theorem]{Conjecture}
\newtheorem{definition}[theorem]{Definition}
\theoremstyle{definition}

\newtheorem{example}[theorem]{Example}
\newtheorem{remark}[theorem]{Remark}

\renewcommand{\epsilon}{\varepsilon}
\newcommand{\N}{\mathbb{N}}
\newcommand{\C}{\mathbb{C}}
\newcommand{\Q}{\mathbb{Q}}
\newcommand{\R}{\mathbb{R}}
\newcommand{\F}{\mc{F}}
\newcommand{\Hr}{{^*\R}}
\newcommand{\Hc}{{^*\C}}

\newcommand{\Ml}{\mathcal{M}_d(\ell^2_\C)}

\newcommand{\beq}{\begin{eqnarray*}}
\newcommand{\eeq}{\end{eqnarray*}}
\newcommand{\be}{\begin{eqnarray}}
\newcommand{\ee}{\end{eqnarray}}
\newcommand{\ben}{\begin{enumerate}}
\newcommand{\een}{\end{enumerate}}

\newcommand{\ba}{\begin{array}}
\newcommand{\ea}{\end{array}}

\newcommand{\ra}{\rangle}
\newcommand{\la}{\langle}
\newcommand{\mc}{\mathcal}

\newcommand{\tr}{\mathrm{tr}}

\newcommand{\eps}{\varepsilon}

\newcommand{\Pp}{\mc{P}}
\newcommand{\Mm}{\mc{M}}

\usepackage[textsize=footnotesize]{todonotes}

\let\originalleft\left
\let\originalright\right
\renewcommand{\left}{\mathopen{}\mathclose\bgroup\originalleft}
\renewcommand{\right}{\aftergroup\egroup\originalright}

\newcommand\xqed[1]{%
  \leavevmode\unskip\penalty9999 \hbox{}\nobreak\hfill
  \quad\hbox{#1}}
\newcommand\demo{\xqed{$\triangle$}}

\usepackage{pgfornament}
\newcommand{\deco}{\raisebox{.6ex}{\pgfornament[width=.4cm,color = gray]{3}}}

\begin{document}

\title{Halos and undecidability of tensor stable positive maps}

\author{Mirte van der Eyden}
\address{Institute for Theoretical Physics, Technikerstr.\ 21a,  A-6020 Innsbruck, Austria}
\email{mirte.van-der-eyden}

\author{Tim Netzer}
\address{Department of Mathematics, Technikerstr.\ 13,  A-6020 Innsbruck, Austria}
\email{tim.netzer}

\author{Gemma De las Cuevas}
\address{Institute for Theoretical Physics, Technikerstr.\ 21a,  A-6020 Innsbruck, Austria}
\email{gemma.delascuevas@uibk.ac.at}

\date{\today}

\begin{abstract}
A map $\mathcal{P}$ is tensor stable positive (tsp) if $\mathcal{P}^{\otimes n}$ is positive for all $n$, and essential tsp if it is not completely positive or completely co-positive. Are there essential tsp maps? Here we prove that there exist essential tsp maps on the hypercomplex numbers. It follows that there exist bound entangled states with a negative partial transpose (NPT) on the hypercomplex, that is, there exists NPT bound entanglement in the halo of quantum states. We also prove that tensor stable positivity on the matrix multiplication tensor is undecidable, and conjecture that tensor stable positivity is undecidable. Proving this conjecture would imply existence of essential tsp maps, and hence of NPT bound entangled states.  
\end{abstract}

\maketitle

%%======================
 We would like to point you to \href{https://youtu.be/G87-Ib2JRfw}{this video}, where this work is presented in an accesible way. 

\section{Introduction} 
\label{sec:intro}
Extremal rays of convex cones play a similar role to basis vectors in vector spaces, 
as they give rise to a description of the cone in terms of positive (instead of linear) combinations thereof. 
The simplest example of a convex cone is that of nonnegative numbers: 
it has  one extremal ray, which gives rise to the finite description   $x\geq 0$. 
The situation for vectors is not much different: 
for vectors from $\R^n$ there is essentially only one notion of nonnegativity, namely that of nonnegative vectors (where every entry is nonnegative), and  they form a convex set with finitely many extreme rays---as many as the size of the vector. 
For matrices, instead, there are \emph{two} main notions of positivity: 
nonnegative matrices (i.e.\ matrices with nonnegative entries) and positive semidefinite matrices (i.e.\ complex Hermitian matrices with nonnegative eigenvalues). 
The first is essentially equivalent to that of nonnegative vectors, in the sense that they form a polyhedron whose extremal rays are the matrices $E_{ij}$ with one element equal to 1 and the rest to 0. 
The second one is fundamentally different:  positive semidefinite matrices form a convex set with \emph{infinitely} many extreme rays. 
They are not only widely studied mathematically, but also at the heart of quantum theory, as they are used to describe quantum states.  

Given an object such as a matrix with a positivity property, it is natural to study maps that preserve that property.  
The natural morphism for positive semidefinite matrices are 
 positivity preserving linear maps, simply called \emph{positive maps}. 
In contrast to completely positive maps, which admit an easy characterisation by Stinespring's Dilation Theorem, 
positive maps are very hard to describe, as they are related to entanglement detection \cite{Gu03}. 
From a mathematical perspective, 
positive maps preserve the cone of positive semidefinite matrices, 
but since this cone does not admit a finite description, neither do the maps.

\begin{figure}[t]\centering
\includegraphics[width = \textwidth]{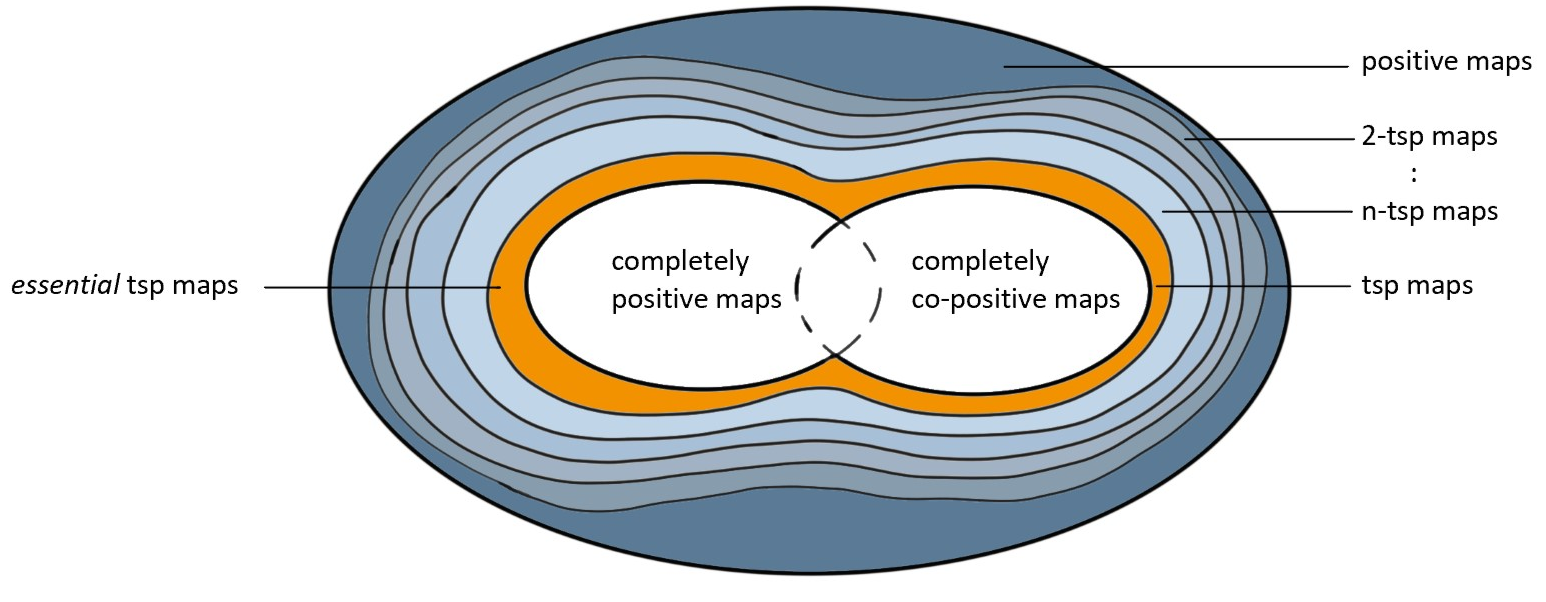}
\caption{\small The set of positive maps with its subsets of $n$-tensor stable positive (tsp) maps. Do there exist \emph{essential} tsp maps, i.e.\ tsp maps that are neither completely positive nor completely co-positive?}
\label{fig:maps}
\end{figure}

One natural composition operation for vector spaces is the tensor product $\otimes$. How does the tensor product interact with the elements of the convex cones mentioned above?  
This is a very rich problem, as the global positivity interacts with the local positivity in highly nontrivial ways \cite{De15,De19d,De20}.  
Here we consider the cone of positive maps, and study the interaction of its elements with the tensor product $\otimes$.  
Specifically, we study which maps stay positive when taking the tensor product with itself an arbitrary number of times. 
Namely, a map  $\mathcal{P}$ is called \emph{tensor stable positive} (tsp) if all its tensor powers are positive, i.e.\ $\mathcal{P}^{\otimes n}$ is positive for all $n$ \cite{Mu16} (see also \cite{Fi16}). 
It is easy to see that if $\mathcal{P}$ is completely positive, or completely positive followed by a transposition (called completely co-positive), 
then  it is  tsp---these are the \textit{trivial} tsp maps \cite{Mu16}. 
But do there exist tsp maps beyond these trivial examples (\cref{fig:maps})? 
In this paper, we call nontrivial tsp maps \emph{essential} tsp maps. 
So the central question is: 
\begin{center}
\emph{Q:\quad Are there essential tsp maps?}\label{Q}
\end{center} 

This question is not only interesting mathematically, 
but is in fact intimately related to a widely studied problem in quantum information theory. 
Namely, if there exist essential tsp maps then there exist non-distillable quantum states with a negative partial transpose (NPT), also called NPT bound entangled states \cite{Mu16}. 
The existence  of NPT bound entanglement has recently been highlighted as one of five important open problems in quantum information theory \cite{Ho20} (see also   \cite{Ho98,Di00,Du00,Ch06,Pa10}). 

In this paper, we approach question \hyperref[Q]{$Q$} from two angles. 
First, we show the existence of essential tsp maps in the field of the \emph{hypercomplex} numbers (\cref{thm:gen}). 
A hypercomplex number is of the form $x+i y$, where $x$ and $y$ are \emph{hyperreal} numbers and $i$ is the imaginary unit, $i^2=-1$. 
The hyperreals are an extension of the reals in which there exist infinitesimal and infinite elements, which are respectively smaller and bigger than any positive real number (\cref{fig:hyperreals}). 
Our result can be intuitively understood as follows: 
the hypercomplex form halos around complex numbers, 
which `glow' outside the set of trivial tsp maps, so there are essential tsp maps living in these halos (\cref{fig:tsp_halo} and \cref{fig:halo}).\footnote{%
At the risk of sounding suspiciously close to quantum mystics, especially regarding the search for an essence in a halo.} 
We call the `quantum' states defined on the hypercomplex  \textit{hyperquantum states}, 
and show that there are NPT bound entangled hyperquantum states (\cref{cor:halos}). In addition, we prove that essential tsp maps exist on  the sequence space $\ell^2$ (\cref{thm:l2}), yet with a very special notion of positivity. 

\begin{figure}[t] \centering
\includegraphics[width=1\columnwidth]{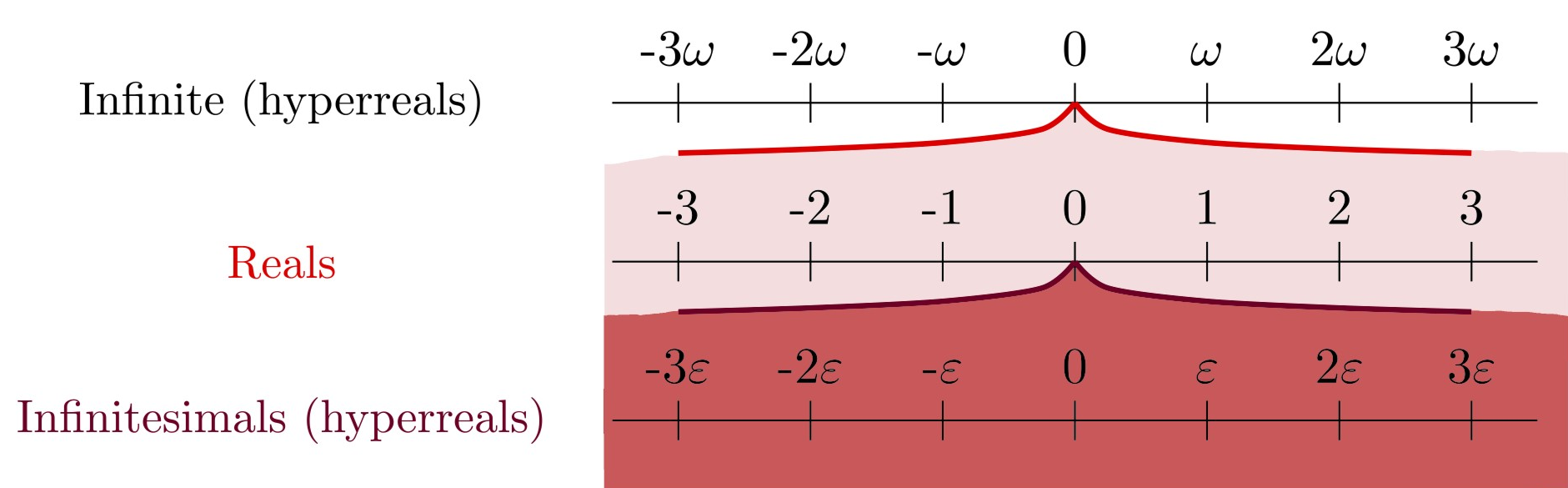}
\caption{\small 
Infinite elements $\omega$ and infinitesimal elements $\epsilon$ in the hyperreals are respectively bigger and smaller than all real numbers.}
\label{fig:hyperreals}
\end{figure}

The second angle concerns computational complexity, in particular \emph{undecidability}. 
While undecidability is well-established in  computer science and mathematics, 
its importance in physics and especially quantum information theory is being explored only recently (see e.g.\ \cite{Wo11, Ei11,Kl14, De15,Cu15b,Sc21} for a sample).  
Here we show that deciding whether a map is tsp on a specific state, namely the matrix multiplication (MaMu) tensor, is undecidable (\cref{thm:main}).  The MaMu tensor is defined as 
\be \label{eq:mamu}
|\chi_n\ra = \sum_{i_1,\ldots,i_n=1}^d 
|i_1,i_2\ra \otimes |i_2,i_3\ra \otimes \cdots \otimes |i_{n},i_1\ra , 
\ee
where $\ket{i}$ denotes the $i$-th vector from the canonical orthonormal basis, 
and $\ket{i_l,i_{l+1}}$ is shorthand for $\ket{i_{l}} \otimes \ket{i_{l+1}}$. 
Our decision problem asks whether all tensor powers of a linear map $\Pp$ send the MaMu tensor to a positive semidefinite matrix, namely: 
\smallskip
\begin{quote} 
Given $d\in \mathbb{N}$ and a linear map  $\mc{P}: \mc{M}_{d^2}  \to\mc{M}_{d^2}  $, \\
is $\mc{P}^{\otimes n} (\ket{\chi_n}\bra{\chi_n})$ positive semidefinite for all $n$? 
\end{quote}
\smallskip
We prove that this problem is undecidable.

\begin{figure}[t]\centering
\includegraphics[width = 0.8\textwidth]{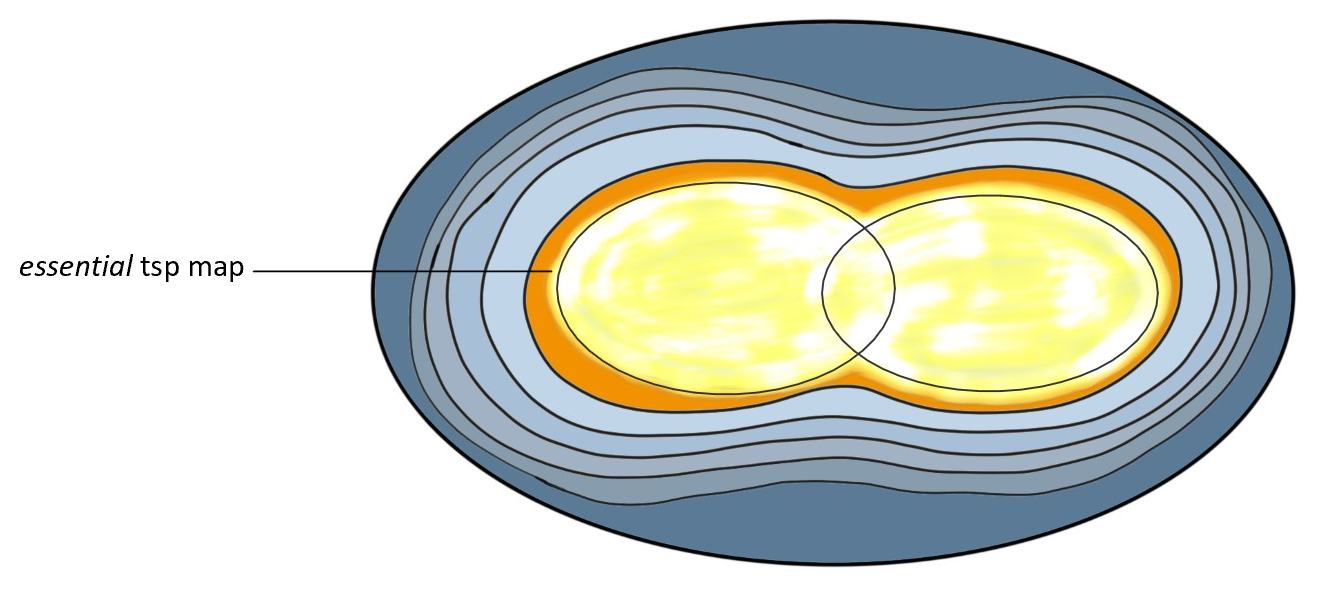}
\caption{\small There are essential tsp maps within the halo of the trivial tsp maps (\cref{thm:gen}), because the latter `glow' in the hypercomplex. (Compare with \cref{fig:maps}). }
\label{fig:tsp_halo}
\end{figure}

This paper is structured as follows. 
In \cref{sec:prelim} we present the basic notions on tensor stable positivity and the hypercomplex field. 
In  \cref{sec:hyper} we prove the existence of essential tsp maps on the hypercomplex field, 
and the existence of NPT bound entangled hyperquantum states.  
In   \cref{sec:undec} we prove the undecidability of tsp maps on the MaMu tensor.
In \cref{sec:concl} we conclude, provide an outlook and discuss the value of our results.  
In Appendix \ref{app:hyper} we give basic properties of the hyperreals,
and in Appendix \ref{app:l2} we reformulate our results on $\ell^2$.

\begin{figure}[t]\centering
\includegraphics[width=.6\columnwidth]{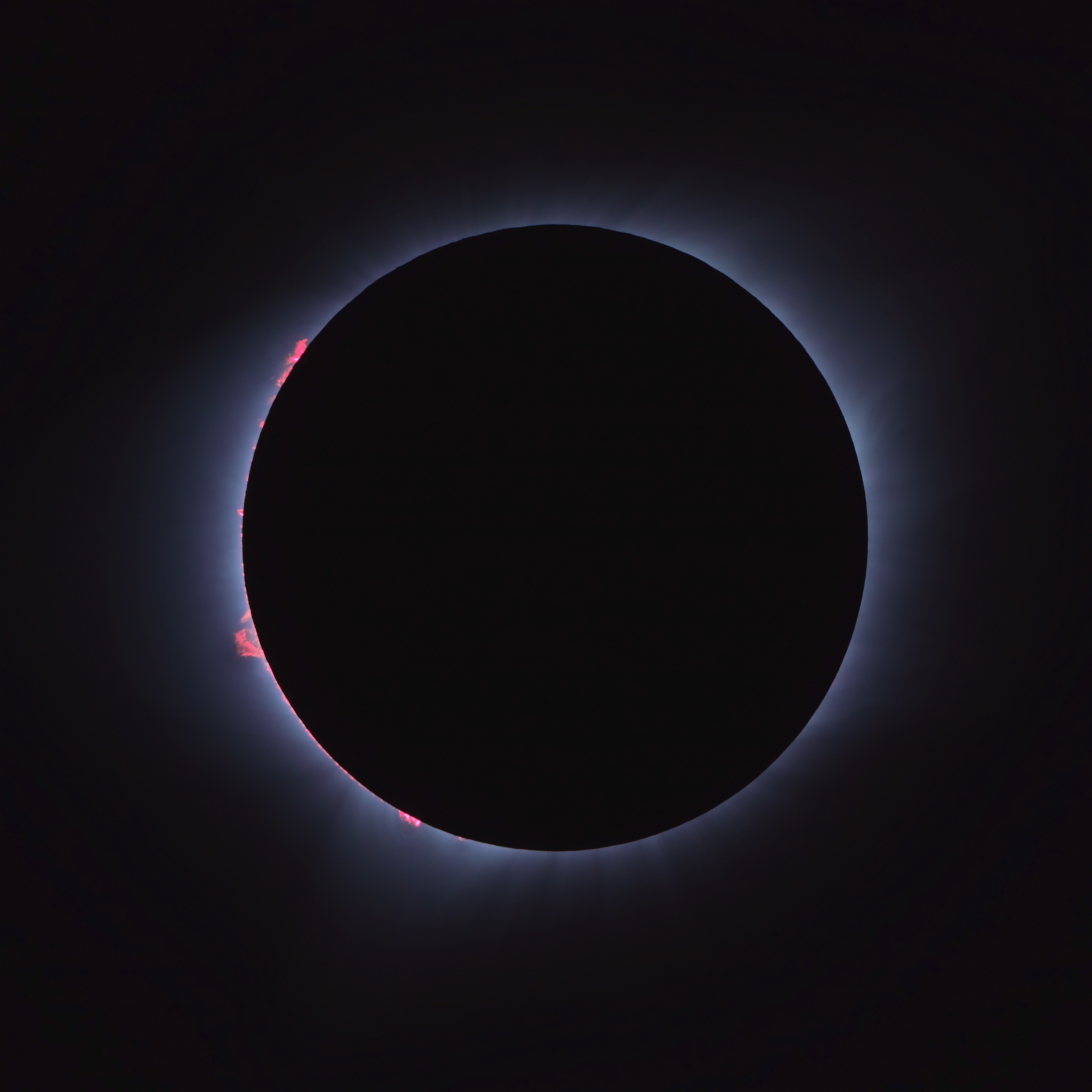}
\caption{{\small How one can imagine a halo of the complex numbers (described by the hypercomplex numbers), containing essential tsp maps. (Photo by Karen Kayser via Unsplash).}}
\label{fig:halo}
\end{figure}

%%====================================
\section{Setting the stage}\label{sec:prelim}

To set the stage  we fix the notation (\cref{ssec:not}), 
give basic properties of tensor stable positivity (\cref{ssec:tsp}) and  of the hypercomplex field (\cref{ssec:hyper_intro}).

%%====================================
\subsection{Notation}
\label{ssec:not}

We denote the computational basis of the Hilbert space $\C^d$ by  $\ket{i}$.\footnote{In mathematics, this is the canonical basis, namely $e_i$ is a column vector containing a 1 in position $i$ and 0 elsewhere.} 
The transposition map with respect to this basis is denoted $\theta(A) := A^T$; 
if $A$ is a $d\times d$ matrix, sometimes we emphasize the dimension of the transposition map as  $\theta_d$. 

The $d \times d$ identity matrix is denoted $\mathbbm{1}_d$, and the set of all $d$-dimensional square matrices with complex entries by $\mc{M}_d$. 
Whenever we consider matrices over a different field or vector space $V$ than the complex numbers $\C$, this will be denoted $\mc{M}_d(V)$.

We  write $\ket{i,j}$, or $\ket{ij}$ when there is no ambiguity, as a shorthand for $\ket{i} \otimes \ket{j}$.
 Given a matrix $A \in \mc{M}_{d_1} \otimes \mc{M}_{d_2} $ with matrix elements given by $\bra{ij}A\ket{kl}$, 
 the partial transpose of the second system, denoted $A^{T_B}$, is defined as  
 $$
 \bra{ij}A^{T_B}\ket{kl} =\bra{il}A\ket{kj}, 
 $$  
 or equivalently 
 $$
 \left(\sum_i X_i \otimes Y_i\right)^{T_B} =  \sum_i X_i \otimes Y_i^T.
 $$
 
Finally,   the \emph{flip operator}   $\mathbb{F}_d: \C^d \otimes \C^d \rightarrow \C^d \otimes \C^d$ acts as $\mathbb{F}_d \ket{ij} = \ket{ji}$.

%%====================================
\subsection{Tensor stable positive maps}
\label{ssec:tsp}
 
Here we define tensor stable positive (tsp) maps. 
First recall that a Hermitian matrix $A\in \mc{M}_d$ is positive semidefinite (psd), denoted $A\geqslant0$, if $\bra{v} A\ket{v} \geq 0$ for all vectors $\ket{v} \in \C^d$, and  $A$ is \emph{separable} if it can be expressed as $A = \sum_{i} \sigma_i\otimes \tau_i$ where all $\sigma_i$ and $\tau_i$ are psd.

For  a linear map 
\be \label{eq:Pcomplex}
\Pp:\mc{M}_{d_1} \rightarrow\mc{M}_{d_2},
\ee
we consider the following ways of preserving the positivity:
\begin{definition}[Notions of positivity]\label{def:pos}\quad
\begin{enumerate}[label=(\roman*),ref=(\roman*)]
\item \label{def:pos:i}
$\Pp$ is \emph{positive}, denoted $\Pp \succcurlyeq 0$,  if it maps psd matrices to psd matrices. 
\item  \label{def:pos:ii}
$\Pp$ is \emph{completely positive} if $\textrm{id}_d \otimes \Pp  \succcurlyeq 0$ for all $d$, where $\textrm{id}_d$ is the identity map on $d\times d$ matrices. 
\item  \label{def:pos:iii}
$\Pp$ is \emph{completely co-positive} if  $\Pp = \theta \circ \mc{S}$ where $\theta$ is the transposition and $\mc{S}$ is a completely positive map. 
\end{enumerate}
The set of positive, completely positive and completely co-positive maps is denoted  $\textrm{POS}$, $\textrm{CP}$ and $\textrm{coCP}$, respectively. 
\end{definition}

Bear in mind that the dimensions $d_1,d_2$ are fixed, despite the fact that our notation for the sets does not make it explicit.

For complete positivity \ref{def:pos:ii}, Choi's Theorem \cite{Ch75} says that the infinite set of conditions defining complete positivity (namely for  all $d\in \N$) is equivalent to a \emph{finite} set conditions, namely
\be \label{eq:Choithm}
\textrm{id}_d \otimes \Pp  \succcurlyeq 0 \quad \textrm{for all } d \leq d_1.  
\ee

Every linear map $\mc{P}$ (Equation \eqref{eq:Pcomplex}) 
can be decomposed as 
\begin{equation}\label{eq:decomp}
\mc{P}(X) = \sum_{i=1}^r A_i \textrm{tr}(B_i^{T}X),
\end{equation} 
where  $B_i \in \mc{M}_{d_1} $ are linearly independent, and so are $A_i \in  \mc{M}_{d_2} $, 
so that $r$ is the rank of the map (i.e.\ the dimension of the image). 
There is a one-to-one correspondence between a linear map $\Pp$ and its Choi matrix 
$$
C_{\Pp} := (\Pp \otimes \textrm{id}_{d_1}) \ket{\Omega}\bra{\Omega}  \quad \textrm{where } 
|\Omega\ra := \frac{1}{\sqrt{d_1}}\sum_{i=1}^{d_1}\ket{ii} , 
$$
where the latter is a maximally entangled state. 
In terms of the decomposition of \eqref{eq:decomp}, 
\be\label{eq:Cpdecomp}
C_{\Pp} = \frac{1}{d_1}\sum_{i=1}^r A_i \otimes B_i \in \mc{M}_{d_2} \otimes \mc{M}_{d_1}. 
\ee
$C_{\Pp}$ is \emph{block positive} if  
$$
(\bra{a} \otimes \bra{b}) C_{\Pp} (\ket{a} \otimes \ket{b})\geq 0 
$$ 
for all vectors $\ket{a}$ and $\ket{b}$. 
The following relations under the Choi-Jamio\l kowski isomorphism are well-known: 
\begin{enumerate}[label=(\alph*)]
\item $\Pp$ is positive iff $C_{\Pp}$ is block positive. 
\item  $\Pp$ is completely positive iff $C_{\Pp}\geqslant 0$.
\item $\Pp$ is completely co-positive iff $C_{\Pp}^{T_B}\geqslant 0$.
\item $\Pp$ is entanglement breaking iff $C_{\Pp}$ is separable. 
\end{enumerate} 
Recall that $C_{\Pp}$ is separable if  $C_{\Pp} = \sum_i^r A_i \otimes B_i$ with $A_i,B_i\geqslant 0 ~\forall i$.   
Note also that if $\Pp$ is positive, then $C_{\Pp}$ is Hermitian and hence $A_i$ and $B_i$ in Equation \eqref{eq:decomp} can be chosen Hermitian too. 

We are now ready to consider tensor products of positive maps  \cite{Mu16}. 
The $n$-fold tensor power of $\Pp$ is given by
\begin{align*}
\mc{P}^{\otimes n}: \mc{M}_{d_1}^{\otimes n} &\to \mc{M}_{d_2}^{\otimes n}  \\
X  \otimes Y \otimes \cdots \otimes Z 
&\mapsto  
\mc{P}(X)  \otimes \mc{P}(Y) \otimes \cdots \otimes \mc{P}(Z) 
\end{align*}
and  this extends to the entire vector space by the linearity of $\mc{P}$. 
Since $\mc{M}_d^{\otimes n } \cong \mc{M}_{d^n}$, 
the map $\Pp^{\otimes n}$ inherits the notion of positivity from $\Pp$, 
namely $\Pp^{\otimes n}$ is positive if it maps psd matrices in $\mc{M}_{d^n}$ to psd matrices. 

\begin{definition}[Tensor stable positivity \cite{Mu16}]\label{def:tsp}
Let  $\Pp:\mc{M}_{d_1} \rightarrow\mc{M}_{d_2}$ be a linear map. 
\begin{enumerate}[label=(\roman*),ref=(\roman*),leftmargin=*]
\item \label{def:tsp:i}
$\Pp$ is \emph{$n$-tensor stable positive} ($n$-tsp) if its $n$-fold tensor product is positive, 
i.e.\ $\Pp^{\otimes n} \succcurlyeq 0$. 
The set of all such maps is denoted $\textrm{TSP}_n$. 
\item  \label{def:tsp:ii}
$\Pp$ is  \emph{tensor stable positive} (tsp) if it is $n$-tsp for all $n$. 
The set of all such maps is denoted $\textrm{TSP}$. 
\end{enumerate}
\end{definition} 

Note that an $n$-tsp map is also $(n-1)$-tsp (see the proof of \cref{lem:geom}). 
In addition, for every $n$ there exists an $n$-tsp map that is not completely positive or completely co-positive \cite{Mu16}. 
Together this shows that there is a nested structure (\cref{fig:maps}):  
$$
\textrm{POS} = \textrm{TSP}_1 \supseteq \textrm{TSP}_2   \supseteq \ldots \supseteq \bigcap_n \textrm{TSP}_n = \textrm{TSP} 
\stackrel{\hyperref[Q]{Q} (?)}{=} \textrm{CP} \cup \textrm{coCP}.
$$ 
It is easy to see that every completely positive and completely co-positive map is tensor stable positive---these are the \emph{trivial} tsp maps. The key question is whether there exist nontrivial, i.e.\ \textit{essential} tsp maps \cite{Mu16}.  % 

A quantum state $\rho \in \mc{M}_{d_1}\otimes  \mc{M}_{d_2}$ is \emph{distillable} if there exists a sequence of maps that can be performed with local operations and classical communication $\Lambda_n$ such that 
$$
\Lambda_n(\rho^{\otimes n}) \to |\Omega\ra\la \Omega|,  
$$ 
and $\rho$ is \emph{bound entangled} if it is not distillable (and it is entangled).  
If $\rho $ is entangled and has a positive partial transpose (PPT), then it is bound entangled   \cite{Ho98}.

%%=================================
\subsection{The hypercomplex field} 
\label{ssec:hyper_intro}

Here we present the basic notions of the hyperreal and the hypercomplex field. For an introduction to the topic we refer to \cite{Go98}. 

As a general rule, all results from linear algebra (like eigenvalues, invertibility, determinants, traces, etc) hold for any field, so in particular to the reals, the hyperreals, the complex and the hypercomplex. 
As we will see, the notions of positivity of \cref{ssec:tsp} can be transferred wholesale to the hyperreals and hypercomplex. This can be seen by applying linear algebra techniques directly to these other fields, or by using the transfer principle (\cref{thm:transfer}).

The \emph{hyperreal} field  $\Hr$  can be defined as the set of infinite sequences of real numbers modulo a certain equivalence relation ${\Hr} =  \R^\N /\sim$ (see Appendix \ref{app:hyper}). 
We can embed any real number $a\in \R$ into the hyperreals by identifying it with the equivalence class of the sequence $(a,a,a,a, \ldots)\in \R^\N$.  
The hyperreals are in many ways similar to $\R$ but contain extra infinitesimal and infinite elements (\cref{fig:hyperreals}), arising for example from sequences that converge to 0 and diverge, respectively. 
Every hyperreal $b$ is surrounded by a set of elements that are infinitely close to $b$ with respect to $\R$. The set of all such elements is  called the \emph{halo} of $b$.\footnote{That's what Beyonc\'e had in mind when writing \href{https://www.youtube.com/watch?v=bnVUHWCynig}{Halo}.}
And conversely: every non-infinite hyperreal $x$ is infinitesimally close to exactly one element $a$ of the real numbers, called the \emph{shadow} (or standard part), denoted $\textrm{sh}(x) = a$.

%%--------------------------------

The \emph{hypercomplex} field $\Hc$ is the complexification  of the hyperreals \cite{Be19}, 
$$
\Hc =  \Hr +   \Hr i, 
$$
where $i^2 = -1$. 
We denote the space of $d \times d$ matrices over $\Hc$ by $\mc{M}_d(\Hc)$. 
A  `quantum state' on the hypercomplex field is described by a matrix $M \in \mc{M}_d(\Hc)$ which is Hermitian, positive semidefinite (i.e.\ with nonnegative eigenvalues) and of trace $1$. We call them \emph{hyperquantum states}:
\begin{center}
Quantum state on the hypercomplex = Hyperquantum state\footnote{At the risk of reminding the reader of \href{https://www.youtube.com/watch?v=7Twnmhe948A}{Hyper Hyper}.}
\end{center}

Some results in  $\Hr$ can be transferred to $\R$, and vice versa. 
Transferring a non-infinite element in the hyperreals to the reals means taking its shadow, 
and  transferring a real $a\in \R$   to the hyperreals means embedding it in $ \Hr$. 
On the other hand, transferring a formula from $\R$ to $\Hr$ means taking its *-transform \cite{Go98}, 
which essentially amounts to replacing $ \forall x \in \R$ by $\forall x \in \Hr$, 
and the same for existential quantifiers. (This is so because the *-transforms of  relations $=,>,<,\neq$ remain the same.) 
The results that can be transferred from the reals to the hyperreals and vice versa are those that can be expressed in \emph{first-order logic}, that is, that involve quantifiers only over the domain of discourse, 
which is $\R$ and $\Hr$ in our case, or $\C$ and $\Hc$ later on: 

\begin{theorem}[Transfer principle \cite{Go98}]\label{thm:transfer}
An $\mc{L}$-sentence $\phi$ is true if and only if its *-transform $^*\phi$ is true. 
\end{theorem}

The symbols $+, \cdot, 0,1$, $\leq$ define the so-called language $\mc{L}$ of ordered rings. 
An $\mc{L}$-sentence is a formal statement that is written with quantifiers, these symbols and finitely many variables. 
The fact that quantifiers run over the whole domain is the defining feature of first-order logic. 

The transfer principle is proven for real closed fields, 
but it can easily be extended to complex versions thereof 
by considering the real and imaginary parts.

All notions of positivity of \cref{ssec:tsp} apply to the hypercomplex too. In particular, for 
 the linear map 
\be \label{eq:P}
\Pp:\mc{M}_{d_1}(\Hc) \rightarrow\mc{M}_{d_2}(\Hc),
\ee
we will consider the same ways of preserving positivity as in \cref{def:pos}.  
Note that Choi's Theorem applies over $\Hc$ because Equation \eqref{eq:Choithm} can be transferred to $\C$, or alternatively because the proof of Choi's Theorem works over $\Hc$. 
The definition of tsp over the hypercomplex is identical to \cref{def:tsp}.

%%===========================
\section{Essential tsp maps in the hypercomplex field}\label{sec:hyper}

Here we prove existence of essential tsp maps in the hypercomplex field (\cref{ssec:non-triv-tsp}) and 
 existence of NPT bound entangled hyperquantum states (\cref{ssec:npt-be}). 
We then reprove  existence of essential tsp maps in the hypercomplex field with a geometric argument, 
and show other geometric properties of the set of tsp maps   (\cref{ssec:geom}). 

%%===========================
\subsection{Essential tsp maps on the hypercomplex}
\label{ssec:non-triv-tsp}

Here we prove the existence of essential tsp maps on the hypercomplex (\cref{thm:gen}). To this end,   
consider a linear map $\Pp$ whose Choi matrix $C_\Pp \in  \mc{M}_{d_1} \otimes \mc{M}_{d_2}$ has the following properties (P):  \label{P}
\begin{enumerate}[label=(P\arabic*), ref=(P\arabic*)]
\item  \label{P1} 
$C_\Pp$ is separable, 
\item \label{P2} 
$C_\Pp$ is rank deficient, and 
\item \label{P3} 
The zero vector is the only product vector in the kernel of $C_\Pp$. 
\end{enumerate}
Then Ref.\  \cite{Mu16} proves the following two statements: 
\begin{enumerate}[label=\arabic*., ref=Statement \arabic*]
\item \label{1}
$(C_\Pp-\epsilon \mathbbm{1})^{\otimes n} $  
is block positive for any 
\be \label{eq:eps}
\epsilon \in [0, \sqrt[n]{||C_\Pp||^n_{\infty}+\mu^n} - ||C_\Pp||_{\infty}],
\ee
where 
\begin{align*} 
\mu = \textrm{min}\{& (\bra{\psi}\otimes \bra{\phi})  C_\Pp(\ket{\psi}\otimes \ket{\phi}) \: \\
& \textrm{ where } \ket{\psi} \in \C^{d_1}, \ket{\phi}\in \C^{d_2},  \braket{\psi|\psi}= \braket{\phi|\phi} = 1\}  
\end{align*}
and $||C_\Pp||_{\infty}$ denotes the operator norm, which is given by the maximal singular value of $C_\Pp$.
\item \label{2}
The matrices $C_\Pp-\epsilon \mathbbm{1}$ and $C_\Pp^{T_B}-\epsilon \mathbbm{1}$ are not psd for any $\epsilon>0$.  
\end{enumerate}
\ref{1} says that for any $n$,  there is an $n$-tsp map, and \ref{2} shows that this $n$-tsp map is essential. 
Together they imply the existence of an essential $n$-tsp map for every $n$.

The following Choi matrix  $C_\Pp \in \mc{M}_{d_1}\otimes \mc{M}_{d_2}$  satisfies \hyperref[P]{(P)}:
\begin{align} \label{eq:example_map} 
C_\Pp = 
(\ket{11} + \ket{22})(\bra{11} + \bra{22}) + \ket{12}\bra{12}+\ket{21}\bra{21} + \sum_{i>2 \textrm{ or } j>2} \ket{ij}\bra{ij} 
\end{align} 
for any $d_1, d_2>2$ \cite{Mu16}. 
For the construction of an essential tsp map on the hypercomplex, we use the following property of the above statements: 

\begin{lemma}[First order logic]
For fixed $n\in \N$, \ref{1} and \ref{2} are   first order sentences in the language of ordered rings (using the entries of $C_\Pp$ as coefficients).  
\label{lem:first_order}
\end{lemma} 

\begin{proof}
For fixed $n$, $d_1$ and $d_2$, \ref{1}  is an expression where one quantifies over sufficiently small $\epsilon$. 
The upper bound for $\epsilon$ is determined by $n$, 
and for each of these $\epsilon$,  block positivity of $ (C_\Pp-\epsilon \mathbbm{1})^{\otimes n}$ can be expressed in terms of quantifiers over the field $\C$. 
\ref{2} is also a first order logic statement, 
since one quantifies over $\epsilon$ and checks positive semidefiniteness. 
\end{proof}

We are now ready to present the first main result of this work: 

\begin{theorem}[Essential tsp maps on the hypercomplex -- First main result]\label{thm:gen}
Let $\eta \in \Hr$ be a positive infinitesimal and  $\mc{P}: \mc{M}_{d_1} \rightarrow \mc{M}_{d_2}$ be a map whose Choi matrix $C_\Pp$ satisfies properties \hyperref[P]{(P)}.
Then the map 
\be \label{eq:Ppeta}
\Pp_{\eta}: \mc{M}_{d_1}(\Hc) \rightarrow \mc{M}_{d_2}(\Hc)
\ee 
with Choi matrix 
\be \label{eq:CPpeta}
C_{\Pp_{\eta}} = C_\Pp-\eta \mathbbm{1} 
\ee
is an essential tsp map on the hypercomplex. 
\end{theorem}

\begin{proof}
By \cref{lem:first_order}, for a fixed $n$, \ref{1} and \ref{2} are  sentences in first order logic, so they can be transferred to the hypercomplex using the transfer principle (\cref{thm:transfer}).

For every $n$ we choose $\epsilon_n$ equal to the upper bound of \eqref{eq:eps}. 
Since every positive infinitesimal $\eta \in \Hr$ satisfies that $\eta < \epsilon_n$ for all $n\in\N$, 
we conclude that the map of Equation \eqref{eq:Ppeta} whose Choi matrix is given by \eqref{eq:CPpeta} 
is $n$-tsp for all $n$. 
It is also essential, since $\eta>0$ and thus both its Choi matrix $C_{\Pp_{\eta}} $ and its partial transpose  have at least one negative eigenvalue by \ref{2}. 
\end{proof}

Note that this proof holds for any real closed field with infinitesimal elements (in particular, it holds for ``smaller" extensions of the reals that also contain infinitesimals).    

Note also that \cref{thm:gen} cannot be stated in first order logic, 
because  to express tensor stable positivity one needs to quantify over all $\N$, which is not the domain of discourse. 
So it cannot be transferred to the complex numbers. 

Moreover, the example of an essential tsp map in the hypercomplex of \cref{thm:gen} becomes a trivial tsp map in the complex. Namely, the shadow of $C_{\Pp_\eta} $ is  $C_\Pp$ (because all infinitesimals are sent to zero), which corresponds to a trivial tsp map, because of \ref{P1}. 

The result of \cref{thm:gen} can also be intuitively understood as follows: 
Over the hypercomplex, there are trivial tsp maps at the boundary of $\textrm{CP} \cup \textrm{coCP}$ that have essential tsp maps in their halo (\cref{fig:tsp_halo}). When transferring back to the complex numbers, all infinitesimals are set to zero and this halo disappears.

Finally, there exist essential tsp maps in $\ell^2$, as we show in \cref{thm:l2} in Appendix \ref{app:l2}. 
In contrast to $\Hc$, $\ell^2$ is a Hilbert space---yet not with our notion of positivity.

%%============================
\subsection{NPT bound entangled hyperquantum states}
\label{ssec:npt-be}

Now we show that there exist NPT bound entangled hyperquantum states (\cref{cor:halos}).  
We start by defining entanglement distillation on $\Hc$. 

\begin{definition}[Entanglement distillation on $\Hc$]\label{def:distill}
A hyperquantum state $\rho$ 
is \emph{distillable}  
if its shadow $\textrm{sh}(\rho)$ is distillable. \footnote{The usual definition of distillability cannot be transferred to the hypercomplex numbers, because the definition involves approximations by converging sequences, and, over the hypercomplex, only sequences that become constant converge.}.
\end{definition}

 By \cite[Theorem 4]{Mu16},  if there exists an essential tsp map $\mc{P}:\mc{M}_{d_1} \rightarrow \mc{M}_{d_2}$, then there exist NPT bound entangled states in $\mc{M}_{d_1} \otimes \mc{M}_{d_1}$ and in $\mc{M}_{d_2} \otimes \mc{M}_{d_2}$. 
Moreover, the proof of \cite[Theorem 4]{Mu16} is constructive and gives a recipe to transform an essential tsp map into an NPT bound entangled state. 
Here we follow this recipe to construct an NPT bound entangled hyperquantum state:

\begin{example}[NPT bound entangled hyperquantum state]\label{ex:npt}
Consider the Choi matrix $C_\Pp$ defined in \eqref{eq:example_map} for $d_1=d_2=3$. 
Since it satisfies \hyperref[P]{(P)}, 
the map $C_{\Pp_\eta}$ of \eqref{eq:CPpeta} is essential tsp on $\Hc$ for infinitesimal $\eta > 0$.  
Following the proof of \cite[Theorem 4]{Mu16} we obtain the matrix 
\beq
A = \sqrt{\frac{3}{2}}\begin{pmatrix} 
0& 1&0\\ 
 -1& 0&0\\
 0& 0&0
\end{pmatrix}   \in \mc{M}_3(\Hc)
\eeq
and define a new Choi matrix $D$ via a so-called local filtering operation, 
\beq
D := (A^\dagger \otimes  \mathbbm{1}_3) C_{\Pp_\eta} ( A  \otimes \mathbbm{1}_3) 
\eeq
where ${}^\dagger$ denotes complex conjugation and transposition. 
It is easy to verify that the partial transpose is not positive, specifically 
\beq
\bra{  \Omega} D^{T_2} \ket{ \Omega} < 0
\eeq
where $|\Omega\ra$ is the hypercomplex maximally entangled state of dimension 3.  

We define the following hyperquantum state as the so-called $U$-twirl of $D$  
$$
\rho  =  \frac{1}{\mathrm{tr}(D)}  \left[
\left( \frac{\mathrm{tr}(D)}{ 8} - \frac{\mathrm{tr}(D  \mathbb{F} _3)}{ 24} \right) 
(\mathbbm{1}_3 \otimes \mathbbm{1}_3) -  
\left( \frac{\mathrm{tr}(D)}{ 24} - \frac{\mathrm{tr}(D \mathbb{F}_3)}{ 8} \right)  \mathbb{F}_3 \right]
$$
resulting in 
\be\label{eq:rho-eta} 
\rho  
 = 
\begin{pmatrix}
 \alpha - \beta &  &&&&&&&   \\
& \alpha && -\beta & &&&& \\
& &\alpha &&& &-\beta && \\
 & -\beta &&\alpha&&&&& \\
&&&&\alpha- \beta&&&& \\
&&&&&\alpha&&- \beta& \\
&&- \beta&&&&\alpha&& \\
&&&&&- \beta&&\alpha& \\
&&&&&&&&\alpha- \beta  \\
\end{pmatrix}
,
\ee
where all unwritten entries are $0$ and where we have defined 
\be \label{eq:alpha}
%\alpha = \frac{1}{  8 } \left(  1  + \frac{\eta}{ 6(1-\eta)}\right)
\alpha = \frac{1}{  8 } \left(  1  + \frac{\eta}{9(1-\eta)}\right)
\quad \textrm{and} \quad  
%\beta = \frac{1}{ 8}\left(\frac{1}{ 3} + \frac{\eta}{ 2(1-\eta)}\right). 
\beta = \frac{1}{ 8}\left(\frac{1}{ 3} + \frac{\eta}{ 3(1-\eta)}\right). 
\ee
Recall that $\alpha,\beta,\eta$ and all matrix entries are hyperreal. 

We claim that $\rho$ of \eqref{eq:rho-eta} is an NPT bound entangled hyperquantum state. 
First, it can be easily checked that $\tr(\rho)=1$. 
Moreover, $\rho$ is psd if $\eta \leq \frac{3}{4}$, which is the case here since it is a positive infinitesimal. 
Furthermore $\rho$ is NPT for
\be \label{eq:eta}
0<\eta<1,
\ee
which is also the case. 
Since its shadow $\textrm{sh}(\rho)$  is PPT (and is thus not distillable), 
$\rho$ is an NPT bound entangled hyperquantum state.
\demo\end{example}

\begin{remark}[NPT bound entangled hyperquantum states via the transfer principle]
The conclusion of \cref{ex:npt} can  be reached via the transfer principle (\cref{thm:transfer}). 
For real $\eta \leq \frac{3}{4}$, 
$\rho$ of \eqref{eq:rho-eta} is psd. 
Transferring  to the hyperreals, $\rho$ is psd for hyperreal $\eta \leq  \frac{3}{4}$, which includes all positive infinitesimals. 
Furthermore, $\rho^{T_B}$ is not psd for $\eta$ in the range \eqref{eq:eta},
which can be transferred as well and shows that $\rho$ is NPT. 
When transferring back to $\C$, the shadow of $\eta$ is 0, 
and we are left with a PPT state. 
\end{remark}

\cref{ex:npt} shows that (\cref{fig:halo}): 

\begin{corollary}[NPT bound entangled hyperquantum state]\label{cor:halos}
There exist NPT bound entangled hyperquantum states. 
\end{corollary}

%%==================================================
\subsection{Geometry of tsp maps}
\label{ssec:geom}
Here we prove again the existence of essential tsp maps over $\Hc$ by a geometric argument. 
To this end, we first investigate the geometry of  the subsets of \textrm{POS} in the complex and hypercomplex 
(\cref{lem:geom}).

For a real closed field $R$ (such as $\R$ and $\Hr$),  
a \emph{semialgebraic set} is a subset of $R^n$ defined by finitely many polynomial equations of the form 
$p(x_1,   \ldots, x_n) \geq 0$ and Boolean combinations thereof. 
By the quantifier elimination theorem \cite{Pr01}, a set defined by a first-order formula in the language of ordered rings (possibly including quantifiers) is semialgebraic. 
We consider sets of Hermitian matrices, which form a real vector space, where they are (potentially) semialgebraic:

\begin{lemma}[Geometry of tsp maps] \label{lem:geom}
The following statements hold   for $\C$ and $\Hc$: 
\begin{enumerate}[label=(\roman*),ref=(\roman*)]
\item  \label{lem:geom:i}
$\textrm{CP}  \cup \textrm{coCP}$  
 is a semialgebraic set.
 \item \label{lem:geom:ii}
$\textrm{TSP}_n$ is semialgebraic for every $n$.

\item \label{lem:geom:iii}

$\textrm{TSP}_n \supseteq \textrm{TSP}_{n+1}$ for all $n\in \N$. 

%$\textrm{TSP}_1 \supseteq \textrm{TSP}_2 \supseteq   \ldots \supseteq \bigcap_n^\infty \textrm{TSP}_n = \textrm{TSP} $.
\end{enumerate}
\end{lemma}

Each of these statements holds for a given, finite size. For example, $\textrm{CP}$ is the set of completely positive maps from $\mc{M}_{d_1}\to\mc{M}_{d_2}$, for fixed $d_1,d_2$. 

\begin{proof}
All arguments hold both for $\C$ and $\Hc$. 

\ref{lem:geom:i}. 
The condition of being completely positive translates to the Choi matrix being psd, 
which can be expressed as a finite set of inequalities in the matrix elements, 
namely that the determinant of every principal minor is nonnegative (Sylvester's criterion). 
Being completely co-positive translates to the partial transpose of the Choi matrix being psd, 
so Sylvester's criterion need only be applied to the partial transpose. 
Finally, the union of two semialgebraic sets $\textrm{CP}$ and $\textrm{coCP}$ is semialgebraic  by the definition of the latter. 

\ref{lem:geom:ii}
For $\Pp \in \textrm{TSP}_n$, the Choi matrix $C_{\Pp^{\otimes n}}$ is block positive, which is a semialgebraic condition on the entries of the Choi matrix (using quantifier elimination).  
Therefore, the set is semialgebraic.

\ref{lem:geom:iii}. 
First note that this holds trivially for the zero map. 
When a positive map $\Pp$ is not the zero map, it can be shown that $\Pp(\mathbbm{1}) \neq 0$.
Consider now a nonzero map $\Pp \in \textrm{TSP}_3$ and a psd matrix   $\sum_i X_i \otimes Y_i \geqslant 0$.  
Then 
\beq 
\left(\Pp \otimes \Pp \otimes \Pp\right)\left(\sum_i X_i \otimes Y_i \otimes  \mathbbm{1} \right) = \left(\sum_i \Pp(X_i) \otimes \Pp(Y_i)\right) \otimes \Pp(\mathbbm{1}) \geqslant 0,
\eeq
because $\sum_i X_i \otimes Y_i \otimes \mathbbm{1} \geqslant 0$. 
Since $\Pp(\mathbbm{1})\geqslant 0$ and $\Pp(\mathbbm{1})\neq 0$,  
then $\sum_i \Pp(X_i) \otimes \Pp(Y_i)\geqslant 0$ and so $\Pp \in \textrm{TSP}_2$. This holds iteratively for any $n$. 
It follows that a map that is in $\textrm{TSP}_n$ is also in $\textrm{TSP}_{n-1}$, i.e.\ the sets have a nested structure. 
\end{proof}

Over $\C$, TSP is an intersection of infinitely many closed sets, and is therefore  a closed set itself. 
Yet, an infinite intersection of semialgebraic sets need  not be semialgebraic, 
so it does not follow from \cref{lem:geom}  \ref{lem:geom:ii} and \ref{lem:geom:iii} that \textrm{TSP} be semialgebraic. 
If \textrm{TSP} were not semialgebraic then essential tsp maps would exist, 
because $\textrm{CP}  \cup \textrm{coCP}$ is semialgebraic by \cref{lem:geom} \ref{lem:geom:i}.

In the following, the set of tsp maps over $\C$ and $\Hc$ is denoted  $\textrm{TSP}$  and  $\textrm{TSP}(\Hc)$, respectively.  

\begin{proposition}[Essential tsp maps on the hypercomplex -- Geometric version]\label{thm:main-geom} 
If $\textrm{TSP}(\Hc)$ is semialgebraic, then there exist essential tsp maps over $\C$.
\end{proposition}

\begin{proof} 
Assume $\textrm{TSP}(\Hc)$ is semialgebraic.
In the hyperreals every countable cover of a semialgebraic set has a finite subcover \cite{Pr01}. 
The set $\textrm{TSP}(\Hc)$ is a countable intersection of semialgebraic sets (by \cref{lem:geom} \ref{lem:geom:iii}), 
\beq
\textrm{TSP}_1 \supseteq \textrm{TSP}_2 \supseteq   \ldots \supseteq \bigcap_n^\infty \textrm{TSP}_n = \textrm{TSP} .
\eeq
This implies that $ \textrm{TSP}_n(\Hc) =  \textrm{TSP}(\Hc)$ for some $n$. 
But for each fixed $k\in\N$, 
\beq
\textrm{TSP}_n(\Hc)= \textrm{TSP}_{n+k}(\Hc)
\eeq
is a statement that transfers to $\C$ by the transfer principle (\cref{thm:transfer}), implying that $\textrm{TSP}_n  = \textrm{TSP}$. 
For every $m$ there is an essential $m$-tsp map \cite{Mu16}, so in particular this holds for $m=n$.  
That is, there exist essential tsp maps over $\C$. 
 \end{proof}

It follows that there are essential tsp maps on the hypercomplex:\footnote{In other words, there is an \emph{essence} in the halo of a trivial tsp map.}
 
\begin{corollary}[Essential tsp maps on the hypercomplex -- Geometric version]\label{cor:main-geom} 
There exist essential tsp maps over $\Hc$. 
\end{corollary}

\begin{proof}
From \cref{thm:main-geom} we conclude that at least one of the following statements is true: there exist essential tsp maps over $\C$ and/or $\textrm{TSP}(\Hc)$ is not semialgebraic. If there exist essential tsp maps over $\C$, these maps can be embedded in  $\Hc$, so there exist essential tsp maps over $\Hc$. If on the other hand $\textrm{TSP}(\Hc)$ is not semialgebraic, then $\textrm{TSP}(\Hc)$ cannot be $\textrm{CP}(\Hc)\cup \textrm{coCP}(\Hc)$, because the latter is semialgebraic. 
\end{proof}

\bigskip
\begin{center}
\deco \:\:  \emph{More on the geometry of tsp maps}\:\:\: \deco
\label{sssec:withoutinv}
\end{center}
\bigskip

Here we examine the geometry of various subsets of the cone of positive linear maps. The following results hold both over $\C$ and $\Hc$.

\begin{proposition}[Convexity]\label{pro:convexity}\quad
\begin{enumerate}[label=(\roman*),ref=(\roman*)]
\item \label{pro:i} 
Both $\textrm{CP}$ and $\textrm{coCP}$ are convex cones. 

\item \label{pro:ii} 
$\textrm{CP} \cap \textrm{coCP}$ is a convex, non-empty cone. 

\item \label{pro:iii} 
$\textrm{CP} \cup \textrm{coCP}$ is not convex. 

\item \label{pro:iv} 
$\textrm{TSP}$ is not convex. 

\item \label{pro:v}
$\textrm{TSP}$ is star convex with respect to every entanglement breaking map $\mc{Q}$. This means that the line segment from $\mc{Q}$ to any point in $\textrm{TSP}$ is contained in $\textrm{TSP}$. 
\end{enumerate}
\end{proposition}

Note that an object shaped like a star (like a starfish) is star convex, but not convex.  The set of trivial tsp maps, that is,  $\textrm{CP}\cup \textrm{coCP}$,  is star convex with respect to any point in the intersection, 
as can be seen in \cref{fig:maps}.

\begin{proof}
%%--------------
\ref{pro:i}. Obvious.

%%------------
\ref{pro:ii}. 
The intersection of convex cones is convex. Furthermore, the intersection is non-empty, 
since every entanglement breaking map is both in $\textrm{CP}$ and $\textrm{coCP}$. 

%%-----------
\ref{pro:iii}. 
Define the map 
$$\gamma: A \mapsto \frac{1}{2} (A + \theta(A)).$$
We claim that $\gamma$ is neither in $\textrm{CP}$ nor in $\textrm{coCP}$. 
Since $\theta \circ \gamma = \gamma$,  
we only have to show that $\gamma \notin \textrm{CP}$. 
It is immediate to verify that $(\textrm{id} \otimes \gamma)   (|\Omega\ra \la \Omega|) \ngeqslant 0$.  

%%----------------
\ref{pro:iv} 
It is easily checked that the map of \ref{pro:iii} is not 2-tsp. 

\ref{pro:v}
Recall that an entanglement breaking map $\mc{Q} \in \textrm{CP} \cap \textrm{coCP}$ and admits a decomposition of the form of Equation \eqref{eq:decomp} with all $A_i,B_i \geqslant 0$. 
We claim that 
$$
\mc{Q}  + \mc{T} \in \textrm{TSP} \: \textrm{ for every }\mc{T}  \in \textrm{TSP}.
$$ 
 By scaling these maps, this shows that every map on the line between $\mc{T}$ and $\mc{Q}$ is tsp, so $\textrm{TSP}$ is a star convex cone. 
 
To prove the claim, start by noting that $(\mc{T} + \mc{Q})^{\otimes n}$ is a sum of maps of the form 
$$
\underbrace{\mc{T} \otimes \cdots \otimes \mc{T}}_{s} \otimes \underbrace{\mc{Q} \otimes \cdots \otimes \mc{Q}}_{r},
$$ 
where $s,r \in \N$ with $s + r = n$, together with all permutations of the tensor factors.  
Applying such a map to a psd input 
$$
\sum_i X^{[1]}_i \otimes  \cdots \otimes X^{[s]}_i \otimes Y^{[1]}_i \otimes  \cdots \otimes Y^{[r]}_i
$$ 
yields
$$
 \sum_{j_1,\ldots,j_r} \mc{T}^{\otimes s}
\left[\sum_i \textrm{tr}(B_{j_1}^T Y_i^{[1]}) \cdots \textrm{tr}(B_{j_r}^T Y_i^{[r]})  X^{[1]}_i \otimes \cdots \otimes X^{[s]}_i \right]
\otimes A_{j_1} \otimes \cdots \otimes A_{j_r}. 
$$
The expression in square brackets is psd because  all $B_i$'s are psd.  
Since $\mc{T}$ is tsp by assumption, $\mc{T}^{\otimes s}$ applied to the expression in square brackets is psd too. 
In addition, all $A_i \geqslant 0$ are psd by assumption. This proves that $\mc{Q}+\mc{T}$ is tsp.
\end{proof}

\cref{pro:convexity} \ref{pro:v} shows that $\textrm{TSP}$ is connected, and even more, that if there is an essential tsp map, it can be found by starting at the boundary of the trivial tsp maps and `walking' a small distance in the direction away from an entanglement breaking map $\mc{Q}$. 
More precisely, it will be of the form 
$$
\mc{B} - \epsilon \mc{Q}, 
$$ 
where $\mc{B}$ is on the boundary of the trivial tsp maps and  $\epsilon>0$ (and is real). 
Taking $\mc{Q}$ as the completely depolarizing map $\mc{Q}(X) = \textrm{tr}(X) \mathbbm{1}$, this is in fact where the examples of $n$-tsp maps from \cite{Mu16} are found. 
Also, when $\epsilon$ is an infinitesimal in the hyperreals, 
the map $\mc{B} - \epsilon \mc{Q}$ is in the halo of $\mc{B}$ (\cref{fig:tsp_star}).  

\begin{figure}[t]\centering
\includegraphics[width = 0.6\textwidth]{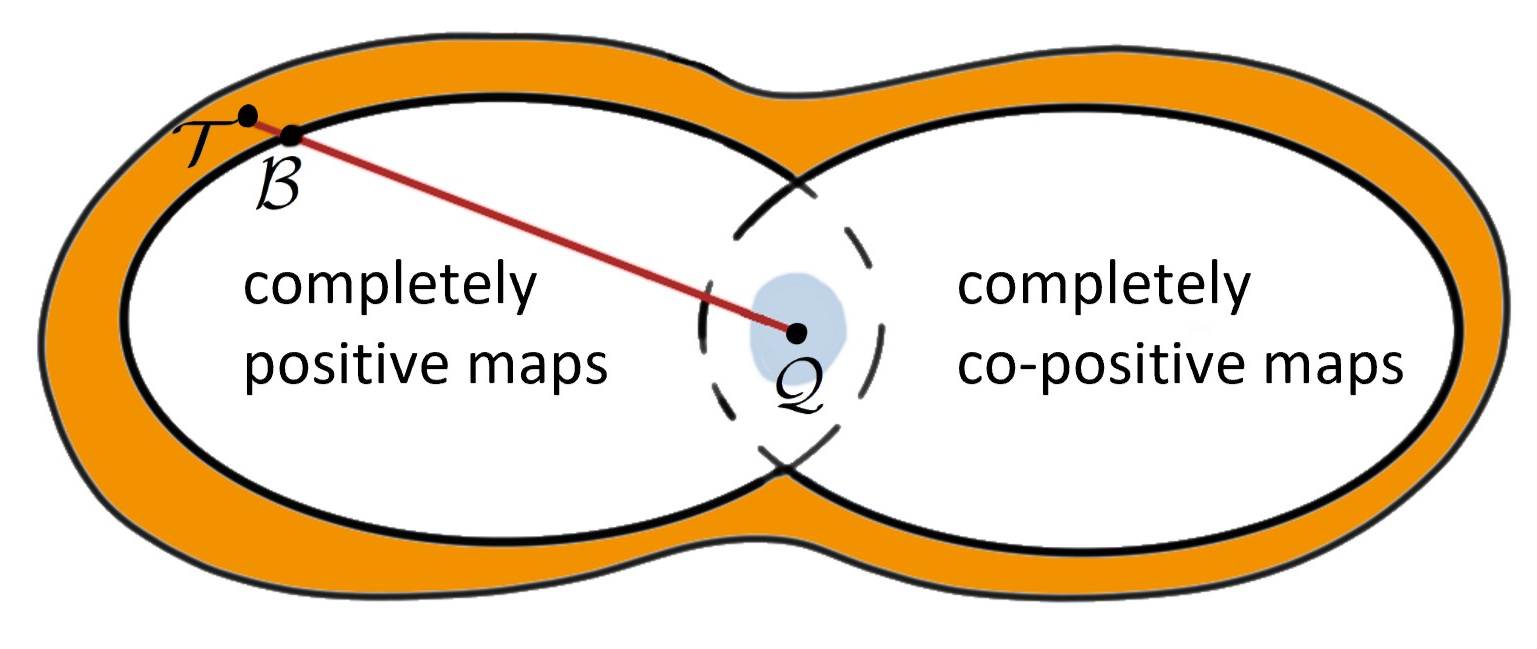}
\caption{\small 
Every tsp map $\mc{T}$ is of the form $\mc{B}-\epsilon \mc{Q}$, where $\mc{B}$ is a map at the boundary of the set of completely positive maps and completely co-positive maps, $\textrm{CP}\cup \textrm{coCP}$, and $\mc{Q}$ is an entanglement breaking map. The set of entanglement breaking maps is colored light blue. 
This follows from the star convexity of $\textrm{TSP}$ with respect to the set of entanglement breaking maps (\cref{pro:convexity} \ref{pro:v}). (Compare with \cref{fig:maps}).   
}
\label{fig:tsp_star}
\end{figure}

%%==========================================
\section{Undecidability of a tensor stable positivity problem}
\label{sec:undec}

Here we prove the  second main result of this work (\cref{thm:main}). 
To this end, we start by defining the decision problem.  
The main actress will be the matrix multiplication (MaMu) state $\ket{\chi_n}$  defined in Equation \eqref{eq:mamu} \cite{Sc81, St86,Ch18}, which is a collection of $d$-dimensional maximally entangled states shared between $n$ neighboring sites (\cref{fig:tensors}). 
More precisely, the main actress will be the \emph{projector} onto the MaMu state: 
$$
\chi_n \coloneqq |\chi_n\ra\la \chi_n| .
$$

%%================FIGURE =========================
\begin{figure}[t]\centering
\begin{tikzpicture}

%%=================(a) ================================
\begin{scope}[yshift = 0cm]

\node(a) at (-5,0) {(a)};
\node (p) at (-2.5,0) {$\ket{\chi_n}=$}; 

\node (a1) at (0,0) {$i_1$};
\node (a2) at (0.4,0) {$i_2$};
\node (a3) at (1.6,0) {$i_2$};
\node (a4) at (2,0) {$i_3$};
\node (a5) at (3.2,0) {$i_3$};
\node (a6) at (3.6,0) {$i_4$};
\node (a7) at (4.8,0) {$i_n$};
\node (a8) at (5.2,0) {$i_1$};

\draw[thick] (a1) -- (0,0.5); 
\draw[thick] (a2) -- (0.4,0.5); 
\draw[thick] (a3) -- (1.6,0.5); 
\draw[thick] (a4) -- (2,0.5); 
\draw[thick] (a5) -- (3.2,0.5); 
\draw[thick] (a6) -- (3.6,0.5); 
\draw[thick] (a7) -- (4.8,0.5); 
\draw[thick] (a8) -- (5.2,0.5); 

\draw[dashed, thick] (-0.4,0.5) --(-0.2,0.5);
\draw[thick] 		(-0.2,0.5) --(0,0.5);
\draw[thick] 		(0.4,0.5) --(1.6,0.5);
\draw[thick] 		(2,0.5) --(3.2,0.5);
\draw[thick] 		(3.6,0.5) -- (3.9,0.5);
\draw[dotted,thick] 	(3.9,0.5) -- (4.5,0.5);
\draw[thick] 		(4.5,0.5) -- (4.8,0.5);
\draw[dashed,thick] 	(5.2,0.5) --(5.5,0.5);
\end{scope}

%%=================(b) ================================
\begin{scope}[yshift = -1.5cm]

\node(b) at (-5,0) {(b)};
    \node (p) at (-2.5 ,-0.5) {$\mc{P}^{\otimes n} = $}; 
    
    \node[draw, shape=rectangle, fill=blue!20, thick] (b0) at (0,0) {$B$};
    \node[draw, shape=rectangle, fill=blue!20, thick] (b1) at (1,0) {$B$};
    \node[draw, shape=rectangle, fill=blue!20, thick] (b2) at (2,0) {$B$};
    \node[draw, shape=rectangle, fill=blue!20, thick] (b3) at (3.5,0) {$B$};
    \node[draw, shape=rectangle, fill=red!20, thick] (a0) at (-0.5,-1) {$A$};
    \node[draw, shape=rectangle, fill=red!20, thick] (a1) at (0.5,-1) {$A$};
    \node[draw, shape=rectangle, fill=red!20, thick] (a2) at (1.5,-1) {$A$};
    \node[draw, shape=rectangle, fill=red!20, thick] (a3) at (3,-1) {$A$}; 
    \draw [thick] (b0) -- (a0)
    (b1) -- (a1)
    (b2) -- (a2)
     (b3) -- (a3);
    \draw [dotted, thick] (2.5,0) -- ( 3.0, 0);
    \draw [dotted, thick] (2.0,-1) -- ( 2.5, -1);
    \draw [thick] (b0) -- (0,0.5) ; 
    \draw [thick] (b1) -- (1,0.5) ;
    \draw [thick] (b2) -- (2,0.5) ;
    \draw [thick] (b3) -- (3.5,0.5) ;
    \draw [thick] (b0) -- (0,-0.5) ; 
    \draw [thick] (b1) -- (1,-0.5) ;
    \draw [thick] (b2) -- (2,-0.5) ;
    \draw [thick] (b3) -- (3.5,-0.5) ;
    \draw [thick] (a0) -- (-0.5,-1.5) ; 
    \draw [thick] (a1) -- (0.5,-1.5) ;
    \draw [thick] (a2) -- (1.5,-1.5) ;
    \draw [thick] (a3) -- (3,-1.5) ;
    \draw [thick] (a0) -- (-0.5,-0.5) ; 
    \draw [thick] (a1) -- (0.5,-0.5) ;
    \draw [thick] (a2) -- (1.5,-0.5) ;
    \draw [thick] (a3) -- (3,-0.5) ;
    \end{scope}

%%=================(c) =========================
\begin{scope}[yshift = -4cm]

\node(c) at (-5,0) {(c)};
    \node (p) at (-2.5,-0.5) {$\mc{P}^{\otimes n}(\chi_n) =$}; 
    
    \draw [thick] (-0.15,0.25) -- (-0.15,0.5) ;     % 0.15 to the side, 0.25 up. (box is 0.5 wide)
    \draw [thick] (0.85,0.25) -- (0.85,0.5) ;
    \draw [thick] (1.85,0.25) -- (1.85,0.5) ;
    \draw [thick] (3.35,0.25) -- (3.35,0.5) ;
    \draw [thick] (0.15,0.25) -- (0.15,0.5) ;     % 0.15 to the side, 0.25 up. (box is 0.5 wide)
    \draw [thick] (1.15,0.25) -- (1.15,0.5) ;
    \draw [thick] (2.15,0.25) -- (2.15,0.5) ;
    \draw [thick] (3.65,0.25) -- (3.65,0.5) ;
  
    \draw [thick] (-0.15,-0.25) -- (-0.15,-0.5) ;     % 0.15 to the side, 0.25 up. (box is 0.5 wide)
    \draw [thick] (0.85,-0.25) -- (0.85,-0.5) ;
    \draw [thick] (1.85,-0.25) -- (1.85,-0.5) ;
    \draw [thick] (3.35,-0.25) -- (3.35,-0.5) ;
    \draw [thick] (0.15,-0.25) -- (0.15,-0.5) ;     % 0.15 to the side, 0.25 up. (box is 0.5 wide)
    \draw [thick] (1.15,-0.25) -- (1.15,-0.5) ;
    \draw [thick] (2.15,-0.25) -- (2.15,-0.5) ;
    \draw [thick] (3.65,-0.25) -- (3.65,-0.5) ;

     \draw [dashed, thick] (-0.25,0.5) -- (-0.55,0.5) ;  
     \draw [thick] (-0.25,0.5) -- (-0.15,0.5) ;  
  \draw [thick] (0.15,0.5) -- (0.85,0.5);     
    \draw [thick] (1.15,0.5) -- (1.85,0.5) ;
    \draw [thick] (2.15,0.5) -- (2.40,0.5) ;
\draw [dotted, thick] (2.40,0.5) -- (3.10,0.5) ;
\draw [thick] (3.10,0.5) -- (3.35,0.5) ;
    \draw [thick] (3.65,0.5) -- (3.75,0.5) ;
    \draw [dashed, thick] (3.75,0.5) -- (4.15,0.5) ;

     \draw [dashed, thick] (-0.25,-0.5) -- (-0.55,-0.5) ;  
     \draw [thick] (-0.25,-0.5) -- (-0.15,-0.5) ;  
  \draw [thick] (0.15,-0.5) -- (0.85,-0.5);     
    \draw [thick] (1.15,-0.5) -- (1.85,-0.5) ;
    \draw [thick] (2.15,-0.5) -- (2.40,-0.5) ;
    \draw [dotted, thick] (2.40,-0.5) -- (3.10,-0.5) ;
    \draw [thick] (3.10,-0.5) -- (3.35,-0.5) ;
    \draw [thick] (3.65,-0.5) -- (3.75,-0.5) ;
    \draw [dashed, thick] (3.75,-0.5) -- (4.15,-0.5) ;

    \node[draw, shape=rectangle, fill=blue!20, thick] (b0) at (0,0) {$B$};
    \node[draw, shape=rectangle, fill=blue!20, thick] (b1) at (1,0) {$B$};
    \node[draw, shape=rectangle, fill=blue!20, thick] (b2) at (2,0) {$B$};
    \node[draw, shape=rectangle, fill=blue!20, thick] (b3) at (3.5,0) {$B$};
    \node[draw, shape=rectangle, fill=red!20, thick] (a0) at (-0.75,-1.5) {$A$};
    \node[draw, shape=rectangle, fill=red!20, thick] (a1) at (0.25,-1.5) {$A$};
    \node[draw, shape=rectangle, fill=red!20, thick] (a2) at (1.25,-1.5) {$A$};
    \node[draw, shape=rectangle, fill=red!20, thick] (a3) at (2.75,-1.5) {$A$}; 
    
    \draw [thick] (b0) -- (a0)
    (b1) -- (a1)
    (b2) -- (a2)
     (b3) -- (a3);
    \draw [dotted, thick] (2.2,-0.5) -- ( 2.6, -0.5);

  \draw [thick] (a0) -- (-0.75,-2) ; 
    \draw [thick] (a1) -- (0.25,-2) ;
    \draw [thick] (a2) -- (1.25,-2) ;
    \draw [thick] (a3) -- (2.75,-2) ;
    \draw [thick] (a0) -- (-0.75,-1) ; 
    \draw [thick] (a1) -- (0.25,-1) ;
    \draw [thick] (a2) -- (1.25,-1) ;
    \draw [thick] (a3) -- (2.75,-1) ;
\end{scope}
\end{tikzpicture}

\caption{{\small Tensor network representations of (a) the MaMu state $\ket{\chi_n}$, 
(b) the $n$-fold tensor product of a linear map $\mc{P} $ decomposed as in \eqref{eq:decomp},
 and  (c) $\mc{P}^{\otimes n}$ applied to the projector on the MaMu state $\chi_n$. 
\textsc{tsp-mamu} asks whether $\mc{P}^{\otimes n}(\chi_n)$ is psd for all $n$.
}}
\label{fig:tensors}
\end{figure}
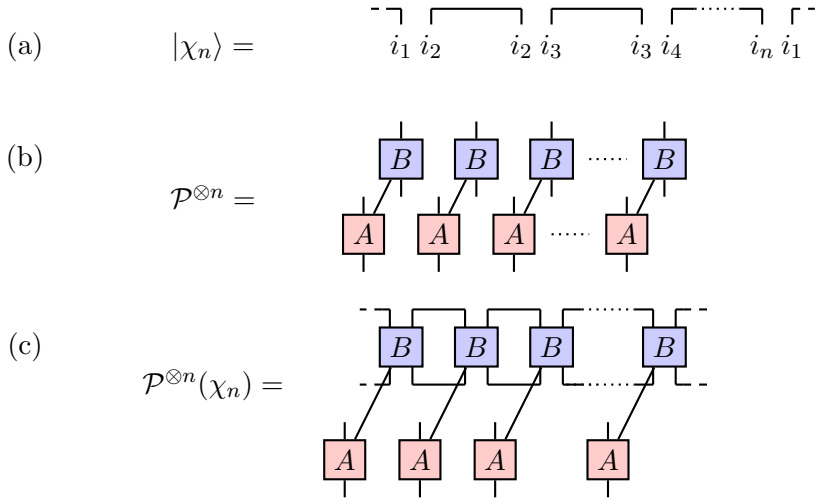

\begin{problem}[\textsc{tsp-mamu}] \label{prob:bell}
Given $d\in \mathbb{N}$ and a linear map  
$\mc{P}: \mc{M}_{d^2} \to\mc{M}_{d^2}$
whose Choi matrix has entries in $\Q$, 
is $\mc{P}^{\otimes n} (\chi_n)\geqslant 0 $ for all $n$? 
\end{problem}

\textsc{tsp-mamu} is asking whether $\mc{P}^{\otimes n}$ maps $\chi_n$ to a psd matrix for all $n$ (\cref{fig:tensors}). 
Note that if $\mc{P}$ is tsp the answer is yes, 
but if $\mc{P}$ is not tsp the answer could still be yes, because we are only `testing'  $\mc{P}^{\otimes n}$  on a specific psd matrix, so $\mc{P}^{\otimes n}$ could fail to be positive on another psd input---which is in fact the case (\cref{rem:failedred}).

The requirement that $C_\mc{P}$ have rational entries ensures that the input of the decision problem is finite. 
In fact, the following results also hold for integer entries,  
since $C_\mc{P}$ can be multiplied by a common multiple of the denominators. 
Since we will prove that \textsc{tsp-mamu} is undecidable, this will also hold for larger input sets, in particular for maps whose Choi matrix has complex entries.

\begin{theorem}[Undecidability of \textsc{tsp-mamu} -- Second main result]\label{thm:main}
\textsc{tsp-mamu} is undecidable, even if $d=3$. 
\end{theorem}

To prove this result, 
we provide a reduction from a problem about Matrix Product Operators (MPO)  \cite{De15}, 
where given a tensor $C = (C^{\alpha,\beta}_i \in \Q) $, one considers the following object: 
\be\label{eq:taunC}
\tau_n(C) \coloneqq \sum_{i_1,\ldots,i_n} 
\tr(C_{i_1}C_{i_2}\cdots C_{i_n}) |i_1\ldots i_n\ra\la i_1 \ldots  i_n| .
\ee

\begin{problem}[\textsc{positive-mpo}]\label{pro:MPO}
Given $s,t\in \N$ and a tensor $C = (C^{\alpha,\beta}_i \in \Q) $ where 
$i \in \{1, \ldots, t\}$ and $ \alpha,\beta \in \{1, \ldots, s\}$,  
is $\tau_n(C)\geqslant 0$ for all $n$?
\end{problem}

\begin{theorem}[Undecidability of \textsc{positive-mpo} \cite{De15}]\label{thm:undec}
\textsc{positive-mpo} is undecidable, even if $s=t=7$. 
\end{theorem}

\begin{proof}[Proof of \cref{thm:main}]
We  provide a computable reduction from \textsc{positive-mpo} with $s,t = 9$ to  \textsc{tsp-mamu}. 
If there would exist an algorithm to solve \textsc{tsp-mamu}, one could use it to decide \textsc{positive-mpo}  (via this reduction), but that contradicts \cref{thm:undec}, so an algorithm for \textsc{tsp-mamu} cannot exist. 

Consider an instance of \textsc{positive-mpo} given by the tensor $C=(C_i^{\alpha, \beta})$ with $(s,t)=(9,9)$. 
We want to show that this is a yes-instance iff its image (under the reduction) is a yes-instance of \textsc{tsp-mamu}. 
So, from $C$, we construct a map $\mc{P}$  as follows (\cref{fig:pictureproof}). 

%%============================
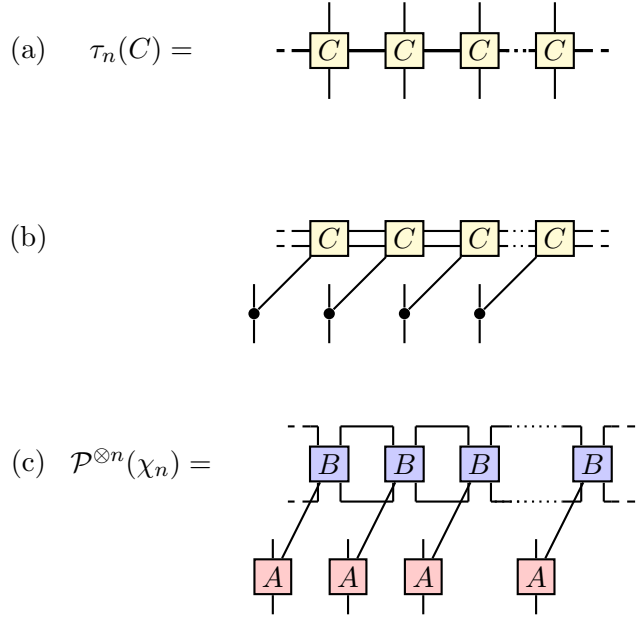
\begin{figure}[t]\centering
\begin{tikzpicture}[inner sep=1mm]
\begin{scope}[xshift = -0.0cm, yshift = -0cm]

\node (a) at (-3,0) {(a)};
\node (p) at (-1.5,0) {$\tau_n(C)=$}; 
   \foreach \i/\itext in {1/1,2/2,3/3,4/n} {
       \node[draw,shape=rectangle, fill=yellow!20, thick] (\i) at (\i, 0) {$C$};
 \node (\i spin) at (\i, -0.75){} ;
        \draw[-,thick] (\i) -- (\i spin);
\node (\i spin) at (\i, 0.75) {};
        \draw[-,thick] (\i) -- (\i spin);
    };
    
    \foreach \i in {1,2} {
        \pgfmathtruncatemacro{\iplusone}{\i + 1};
        \draw[-,very thick] (\i) -- (\iplusone);
    };

\draw[dashed, very thick] (0.3,0) -- (0.5,0);
\draw[-, very thick] (0.5,0) -- (0.75,0);
\draw[-, very thick] (1.25,0) -- (1.75,0);
\draw[-, very thick] (2.25,0) -- (2.75,0);
\draw [very thick]( 3.25,0) -- (3.35,0);
\draw[dotted, very thick] (3.35,0) -- (3.65,0);
\draw [very thick ](3.65,0) -- (3.75,0);
\draw[dotted, very thick] (3.35,-0) -- (3.65,-0);
\draw[-, very thick] (4.25,-0) -- (4.4,-0);
\draw[dashed, very thick]  (4.4,0) -- (4.8,0);

\end{scope}

%%=======(b)==================================
\begin{scope}[xshift = -0.0cm, yshift = -2.5cm]
\node (b) at (-3,0) {(b)};
   \foreach \i/\itext in {1/1,2/2,3/3,4/N} {
       \node[draw,shape=rectangle, fill=yellow!20, thick] (\i) at (\i, 0) {$C$};
    };

\draw[dashed, thick] (0.3,-0.1) -- (0.5,-0.1);
\draw[dashed, thick] (0.3,0.1) -- (0.5,0.1);
\draw[-, thick] (0.5,-0.1) -- (0.75,-0.1);
\draw[-, thick] (0.5,0.1) -- (0.75,0.1);
\draw[-, thick] (1.25,0.1) -- (1.75,0.1);
\draw[-, thick] (1.25,-0.1) -- (1.75,-0.1);
\draw[-, thick] (2.25,-0.1) -- (2.75,-0.1);
\draw[-, thick] (2.25,0.1) -- (2.75,0.1);
\draw[thick] (3.25,0.1) -- (3.35,0.1);
\draw[dotted, thick] (3.35,0.1) -- (3.65,0.1);
\draw[thick] (3.65,0.1) -- (3.75,0.1);
\draw[thick] (3.25,-0.1) -- (3.35,-0.1);
\draw[dotted, thick] (3.35,-0.1) -- (3.65,-0.1);
\draw[thick] (3.65,-0.1) -- (3.75,-0.1);
\draw[-, thick] (4.25,-0.1) -- (4.4,-0.1);
\draw[-, thick] (4.250,0.1) -- (4.4,0.1);
\draw[dashed, thick]  (4.4,-0.1) -- (4.8,-0.1);
\draw[dashed, thick] (4.4,0.1) -- (4.8,0.1);

   \foreach \a/\atext/\ai in {0/1/1,1/2/2,2/3/3,3/n/4} {
	\node[circle,fill,inner sep=1.5pt] (\a dot) at (\a,-1) {};
   \node (\a spin) at (\a, -1.5) {}; 
    \draw[-, thick] (\a dot) -- (\a spin);
  \node (\a spin) at (\a, -0.5) {} ;
    \draw[-, thick] (\a dot) -- (\a spin);
    \draw[-, thick] (\a dot) -- (\ai);
    };

%%=======(c)==================================
\end{scope}
\begin{scope}[shift = {(0,-5.5)}]
\node (c) at (-3,0) {(c)};
\end{scope}

\begin{scope}[shift = {(1,-5.5)}]
   \node (p) at (-2.5,0) {$\mc{P}^{\otimes n}(\chi_n) =$}; 
    \node[draw, shape=rectangle, fill=blue!20, thick] (b0) at (0,0) {$B$};
    \node[draw, shape=rectangle, fill=blue!20, thick] (b1) at (1,0) {$B$};
    \node[draw, shape=rectangle, fill=blue!20, thick] (b2) at (2,0) {$B$};
    \node[draw, shape=rectangle, fill=blue!20, thick] (b3) at (3.5,0) {$B$};
    \node[draw, shape=rectangle, fill=red!20, thick] (a0) at (-0.75,-1.5) {$A$};
    \node[draw, shape=rectangle, fill=red!20, thick] (a1) at (0.25,-1.5) {$A$};
    \node[draw, shape=rectangle, fill=red!20, thick] (a2) at (1.25,-1.5) {$A$};
    \node[draw, shape=rectangle, fill=red!20, thick] (a3) at (2.75,-1.5) {$A$}; 
    \draw [thick] (b0) -- (a0)
    (b1) -- (a1)
    (b2) -- (a2)
     (b3) -- (a3);
    \draw [dotted, thick] (2.2,-0.5) -- ( 2.6, -0.5);
   
    \draw [thick] (-0.15,0.25) -- (-0.15,0.5) ;     
    \draw [thick] (0.85,0.25) -- (0.85,0.5) ;
    \draw [thick] (1.85,0.25) -- (1.85,0.5) ;
    \draw [thick] (3.35,0.25) -- (3.35,0.5) ;
    \draw [thick] (0.15,0.25) -- (0.15,0.5) ;    
    \draw [thick] (1.15,0.25) -- (1.15,0.5) ;
    \draw [thick] (2.15,0.25) -- (2.15,0.5) ;
    \draw [thick] (3.65,0.25) -- (3.65,0.5) ;
  
    \draw [thick] (-0.15,-0.25) -- (-0.15,-0.5) ;     
    \draw [thick] (0.85,-0.25) -- (0.85,-0.5) ;
    \draw [thick] (1.85,-0.25) -- (1.85,-0.5) ;
    \draw [thick] (3.35,-0.25) -- (3.35,-0.5) ;
    \draw [thick] (0.15,-0.25) -- (0.15,-0.5) ;     
    \draw [thick] (1.15,-0.25) -- (1.15,-0.5) ;
    \draw [thick] (2.15,-0.25) -- (2.15,-0.5) ;
    \draw [thick] (3.65,-0.25) -- (3.65,-0.5) ;
  
     \draw [dashed, thick] (-0.25,0.5) -- (-0.55,0.5) ;  
     \draw [thick] (-0.25,0.5) -- (-0.15,0.5) ;  
  \draw [thick] (0.15,0.5) -- (0.85,0.5);     
    \draw [thick] (1.15,0.5) -- (1.85,0.5) ;
    \draw [thick] (2.15,0.5) -- (2.40,0.5) ;
\draw [dotted, thick] (2.40,0.5) -- (3.10,0.5) ;
\draw [thick] (3.10,0.5) -- (3.35,0.5) ;
    \draw [thick] (3.65,0.5) -- (3.75,0.5) ;
    \draw [dashed, thick] (3.75,0.5) -- (4.15,0.5) ;

     \draw [dashed, thick] (-0.25,-0.5) -- (-0.55,-0.5) ;  
     \draw [thick] (-0.25,-0.5) -- (-0.15,-0.5) ;  
  \draw [thick] (0.15,-0.5) -- (0.85,-0.5);     
    \draw [thick] (1.15,-0.5) -- (1.85,-0.5) ;
    \draw [thick] (2.15,-0.5) -- (2.40,-0.5) ;
    \draw [dotted, thick] (2.40,-0.5) -- (3.10,-0.5) ;
    \draw [thick] (3.10,-0.5) -- (3.35,-0.5) ;
    \draw [thick] (3.65,-0.5) -- (3.75,-0.5) ;
    \draw [dashed, thick] (3.75,-0.5) -- (4.15,-0.5) ;

  \draw [thick] (a0) -- (-0.75,-2) ; 
    \draw [thick] (a1) -- (0.25,-2) ;
    \draw [thick] (a2) -- (1.25,-2) ;
    \draw [thick] (a3) -- (2.75,-2) ;
    \draw [thick] (a0) -- (-0.75,-1) ; 
    \draw [thick] (a1) -- (0.25,-1) ;
    \draw [thick] (a2) -- (1.25,-1) ;
    \draw [thick] (a3) -- (2.75,-1) ;
\end{scope}
\end{tikzpicture}

\caption{{\small The reduction from  \textsc{positive-mpo} to   \textsc{tsp-mamu} in tensor network diagrams.  
(a) $\tau_n(C)$, given by Equation \eqref{eq:taunC}, is the central object in \textsc{positive-mpo}.
(b) We `split' the fat horizontal index of $C$ into two indices (Equation \eqref{eq:fat}), 
and express the vertical indices with a delta function (indicated with a dot).  
(c) We identify $A$ with the delta function, 
 and  $B$ with a reshuffled $C$ as given in Equation \eqref{eq:Cgamma}. 
 The diagram of (c) equals the diagram of (b), which equals the diagram of (a), as given in Equation \eqref{eq:final}.
}}
\label{fig:pictureproof}
\end{figure}

We choose $d=3$, $r = t$ and 
$$
A_i =\ket{i}\bra{i}. 
$$
Since $\alpha,\beta$ run to $d^2=s$, we can express each as a multiindex, 
\be \label{eq:fat}
\alpha=(\mu,\nu), \quad \beta=(\lambda,\rho), 
\ee
 where $\mu,\nu,\lambda,\rho=1,\ldots, d$, and we  define the tensor $B$ as 
\be
B_i^{(\mu,\lambda),(\nu,\rho)} = C_i^{(\mu,\nu),(\lambda,\rho)}. 
\label{eq:Cgamma}
\ee
It is now immediate to see that 
\begin{align*}
\bra{\chi_n}  (B_{i_1} \otimes B_{i_2}  \otimes \ldots \otimes B_{i_n} ) \ket{\chi_n} 
&=
 \tr(C_{i_1} C_{i_2} \cdots C_{i_n}). 
\end{align*}
Since the Choi matrix of the $n$-fold tensor product of $\mc{P}$ is given by 
$$
C_{\Pp}^{\otimes n}  = \sum_{i_1,\ldots, i_n=1}^r (A_{i_1} \otimes  \cdots \otimes A_{i_n} )\otimes 
(B_{i_1}   \otimes \cdots \otimes B_{i_n} ),  
$$
we obtain that 
$$
\tau_n(C) = (\mathbbm{1}\otimes  \bra{\chi_n}) C_{\Pp}^{\otimes n} (\mathbbm{1}\otimes \ket{\chi_n} ) .  
$$
In other words (or in other symbols): 
\be \label{eq:final}
\tau_n(C) = \mc{P}^{\otimes n}(\chi_n). 
\ee
Since they are the same,   
the left hand side is psd iff the right hand side is psd (for every $n$), 
which proves the reduction from \textsc{positive-mpo} to \textsc{tsp-mamu}. 
\end{proof}

(Un)fortunately, this does not immediately imply undecidability of tsp: 

\begin{problem}[\textsc{tsp}]\label{prob:tsp}
Given $d\in \mathbb{N}$ and a linear map  $\mc{P}: \mc{M}_d  \to\mc{M}_d  $ whose Choi matrix has entries in $\Q$, is $\mc{P}^{\otimes n}$ positive for all $n$? 
\end{problem}

\begin{remark}[\textsc{tsp-mamu} cannot be reduced to \textsc{tsp} in the obvious way]\label{rem:failedred}
`The obvious way' is the identity map---we show that the identity map from \textsc{tsp-mamu} to \textsc{tsp} is not a reduction. There could exist another reduction, though. 

A reduction maps yes-instances to yes-instances, and no-instances to no-instances. 
The following map  $\mc{P}$ is a yes-instance of \textsc{tsp-mamu} and a no-instance of \textsc{tsp}.
In decomposition of  Equation \eqref{eq:decomp}, it is given by  $r=1$ and 
$$
A = \mathbbm{1}_{d^2}, \quad B = \textrm{diag}(-1,0,\ldots,0,2) \in \mc{M}_{d^2}. 
$$ 
$\mc{P}$ is  not a positive map, because  
$$
P(|1\ra\la 1|) = -\mathbbm{1}.
$$ 
Yet, 
$$
\Pp^{\otimes n}(\chi_n) = 
\mathbbm{1}^{\otimes n}  \tr(B^n) = 
\mathbbm{1}^{\otimes n} [(-1)^n + 2^n] 
$$ 
which is psd for all $n$. 
\demo \end{remark}

A problem\footnote{More precisely, the set of yes-instances of this problem, which defines a formal language.} is recursively enumerable (\textsf{r.e.})\ if it is recognised by a Turing machine, and 
co-recursively enumerable (\textsf{co-r.e.}) if its complement is \textsf{r.e.} \cite{Ko97}. 
A problem is semidecidable  if it is  \textsf{r.e.}\ or \textsf{co-r.e.}, 
and decidable  if it is \textsf{r.e.}\ and \textsf{co-r.e.} (that is, there is a Turing machine that accepts all yes-instances and rejects all no-instances). 

\begin{remark}[Semidecidability of \textsc{TSP}] 
\textsc{tsp} is  \textsf{co-r.e.}, 
because the no-instances can be recognised by a Turing machine (and we conjecture that the yes-instances cannot, cf.\ \cref{con:tsp}). 
Starting from $n=1$ and increasing in $n$, this Turing machine checks whether $\Pp^{\otimes n}$  is a positive map. 
Checking positivity of the map is computable, because of the quantifier elimination theorem. 
If a map $\Pp$ is not tsp, then there is an $n \in \N$ such that $\Pp^{\otimes n}$ is not positive, so the algorithm will find it in finite time and reject the instance. 
If a map  $\Pp$  is tsp, this algorithm will not halt. 
\demo \end{remark}

\begin{conjecture}[\textsc{tsp} is undecidable]\label{con:tsp}
\textsc{tsp} is not \textsf{r.e.}
\end{conjecture}

If \textsc{tsp} were undecidable, essential tsp maps (over the complex) would exist. 
This is so because checking whether a map is in CP or coCP is decidable, so  
if all tsp maps were trivial, an algorithm to decide \textsc{tsp} would exist. 
More precisely, the undecidability of \textsc{tsp} would entail 
\begin{enumerate} 
\item the existence of  essential  tsp maps, 
\item the existence of NPT bound entangled states, and
\item disprove the PPT squared conjecture \cite{Mu18}. 
\end{enumerate}
Yet, these implications are non-constructive, meaning that even if we know that essential tsp maps exist,  
we may not be able to construct one. 
From a broader perspective, undecidability would be a means to proving the existence of essential tsp maps, that is, it would be a proof technique, and not an end in itself. This is already the case for  the undecidability of \textsc{positive-mpo} (\cref{thm:undec}), which is a proof technique to conclude that certain purifications cannot exist \cite{De15}. 

%%===============================
\section{Conclusion and outlook}\label{sec:concl}

In this paper, we have approached the existence of essential tsp maps \cite{Mu16} from two angles. 
First, we showed that essential tsp maps exist on the hypercomplex field (\cref{thm:gen}) and on $\ell_{\C}^2$ (\cref{thm:l2}), 
and that bound entangled hyperquantum states with a negative partial transpose exist (\cref{cor:halos}). 
Second, we proved the undecidability of the tensor stable positivity problem on MaMu tensors (\cref{thm:main}).

One question overlooking this work is whether  tensor stable positivity is undecidable (\cref{con:tsp}), which is part of a bigger trend of exploring the scope of undecidability in physics (see also \cite{De20d}). 
Often, when a problem is undecidable, a bounded version thereof is \textsf{NP}-complete---this is the case for the (bounded) halting problem, the (bounded) post correspondence problem, the (bounded) tiling problem, the (bounded) matrix mortality problem, and the (bounded) \textsc{positive-mpo} (\cref{pro:MPO}) \cite{Kl14}, to cite a few. 
A bounded version of \textsc{tsp} could be \textsf{NP}-complete. We are currently investigating this direction.

 How valuable is it to prove that essential tsp maps exist \emph{on the hypercomplex}?
There are many investigations regarding the `border' of quantum mechanics. 
For example, generalised probabilistic theories try to single out quantum mechanics from a more general set of theories.
Similarly,  reconstructions of quantum mechanics aim at providing physically motivated postulates for quantum mechanics \cite{Ch19b}. 
The hypercomplex are not part of any `orthodox' formulation of quantum mechanics (as far as we know), but this paper shows that some long-standing problems (like the existence of NPT bound entanglement) become solvable there. 
How reasonable is it to assume that our physical reality is in some way described by hypercomplex numbers? 
Clearly, not very reasonable at all, but neither is the assumption that our reality is described by objects requiring an infinite description, such as the reals or complex (see the recent works \cite{Gi18,De19e,Gi19}). 
On the other hand, recent work highlights the need of complex numbers in quantum theory \cite{Re21} (or more precisely, the need for numbers with a real and an imaginary part), and when complex numbers were invented, who would have thought that the square root of $-1$ would be of any use, let alone be necessary, 
for the formulation of a fundamental theory of our world, namely quantum mechanics?

A downside of our result on essential tsp maps on $\mathcal{M}_d(\Hc)$ (\cref{thm:gen}) is that 
$\mathcal{M}_d(\Hc)$ is not a Hilbert space \footnote{The notion of a Hilbert space is only defined over the real or complex numbers. One could relax this condition and try to define a Hilbert space over $\Hc$, but would again run into the problem that over $\Hc$ only constant sequences converge. This problem arises when imposing completeness (with respect to the norm induced by the inner product), which is one of the properties of a Hilbert space. }. 
For this reason we attempted to reformulate our result in $\ell_\C^2$ (Appendix \ref{app:l2}),
 but the notion of positivity there clashes with the existence of an inner product (\cref{pro:inner}), so  the resulting space is not a Hilbert space either. 
 So \cref{thm:gen} is not only challenging because it uses an unorthodox field, namely $\Hc$, 
 but also because the space where these positive maps live is not a Hilbert space---so both aspects challenge the standard formulation of quantum mechanics.

How valuable is it to prove that a `physical' problem (like tsp on MaMu) is undecidable? 
If one disagrees with the use of infinities in physics, then all problems become decidable, because undecidability requires an infinite number of instances.\footnote{This is ultimately due to the fundamental distinction between finite and infinite made in computer science and formal systems, which seems irrelevant for physical quantities, since   any number larger than, say, a googol, $10^{100}$,   is `practically' infinite.} 
Yet, the undecidability of tensor stable positivity would be a non-constructive  \emph{proof technique},  as emphasized at the end of \cref{sec:undec}. 
In this respect, proving undecidability would be useful even if one distrusts objects involving infinities---however, the very definition of tsp involves an infinity (namely for all $n$), so in this case one would disregard the entire question and work. 

\bigskip

%%======================================================
\emph{Acknowledgements}.---We thank Alexander M\"uller-Hermes for pleasant discussions and sharing with us ideas about tensor stable positivity. 
We also thank everyone in the group (Andreas Klingler, Joshua Graf, Tobias Reinhart and Sebastian Stengele) for many discussions. 
MVDE acknowledges support of the Stand Alone Project P33122-N of the Austrian Science Fund (FWF). 

\bigskip 

%%======================================================
\renewcommand{\thesubsection}{\Alph{section}.\arabic{subsection}}
\setcounter{section}{0}

\begin{appendices}

%%======================================================
\section{The hyperreals}
\label{app:hyper}

Here we construct the hyperreals via the ultrapower construction and give an example of an infinitesimal element in the field. This material is based on \cite{Go98}. 

Consider the set $\R^\N$ of all sequences of real numbers. 
An element in this set is of the form $r = (r_1, r_2, r_3, \ldots )$, which we denote by $(r_n)$. 
Defining addition and multiplication entrywise, 
\begin{eqnarray}
\nonumber r + s &=  (r_n + s_n : n\in \N) \\
\nonumber r \cdot s &=  ( r_n \cdot s_n: n \in \N), 
\end{eqnarray}
we obtain that the set  $\R^\N$ is a commutative ring. 
The real numbers can be included in $\R^\N$ by assigning to $a\in \R$ the element 
$ (a,a,a,\ldots)$.   
The zero element of the set is then $ (0,0,0,\ldots)$ and the unity $(1,1,1,\ldots)$. 
Finally, the additive inverse is given by $-r = (-r_n)$. 
Yet, $(\R^\N, +,\cdot)$ is  not a field, because there exist non-zero zero divisors such as 
$$
(1,0,1,0,1,0,\ldots) \cdot (0,1,0,1,0,1,\ldots) = (0,0,0,\ldots).
$$
To `fix' this, and construct a field $\Hr$ out of this ring, one of the previous elements needs to be   0 in $\Hr$. 
This is formalized by means of an ultrafilter $\mc{F}$: 
Two sequences are \emph{equivalent} if the indices for which they are equal form a `large' subset of $\N$, that is, these indices are in $\mc{F}$.

\begin{definition}[Ultrafilter]\label{def:ultrafilter}
An \emph{ultrafilter} $\mc{F}$ on $\N$ is a set of subsets of $\N$ with the following properties: 
\begin{enumerate}
    \item It is closed under taking supersets: if $X\in \mc{F}$ and $X\subseteq Y \subseteq \N$, then $Y \in \mc{F}$. 
    \item It is closed under intersections: if $X,Y\in \mc{F}$, then $X\cap Y \in \mc{F}$. 
    \item $\N\in \mc{F}$ and $\emptyset \notin \mc{F}$. 
    \item For every subset $X \in \N$, exactly one of $X$ and $\N \setminus X$ is in $\mc{F}$. 
\end{enumerate}
\end{definition}

An ultrafilter is called \emph{nonprincipal} (or `free') if it contains no finite subset of $\N$, and therefore all cofinite subsets of $\N$ (see \cref{ex:ultrafilter} for an example of a principal and nonprincipal ultrafilter). This  type of ultrafilter will be used to construct the equivalence relation. It can be proven that every infinite set has a nonprincipal ultrafilter on it. 

Given a nonprincipal ultrafilter $\mc{F}$  on $\N$, the equivalence relation $\sim$ on $\R^\N$  is defined as follows: 
$$
(r_n) \sim (s_n) ~~\Leftrightarrow ~~\{n\in \N: r_n = s_n\} \in \mc{F}.
$$
In words, this relation says that two sequences are equivalent if they are the same on a large set of indices obeying some nice conditions.  
 The equivalence class $[r]$ of a sequence $r\in \R^\N$   is given by 
$$
[r] = \{s \in \R^\N: r \sim s\}.
$$
The hyperreals $\Hr$ are defined as the quotient set 
$$ 
{\Hr} :=  \R^\N /\sim~ = \{[r]: r \in \R^\N\}. 
$$

The hyperreals $\Hr$ together with addition ($[r]+[s] = [(r_n + s_n)]$), multiplication ($[r]\cdot[s] = [(r_n\cdot s_n)]$) and the order relation
$$
[r] < [s] ~~\textrm{iff}~~ \{n\in \N: r_n < s_n\} \in \mc{F}, 
$$ 
is an ordered field. 

In $\Hr$ there are \emph{infinitesimal elements}, that is, elements that are positive but smaller than all positive real numbers.  Their multiplicative inverses are \emph{infinitely large}. 
Let us consider as an example the sequence 
$$
\epsilon = (1, \tfrac{1}{2}, \tfrac{1}{3},\tfrac{1}{4}, \ldots) = (\tfrac{1}{n}).
$$ 
The element $[\epsilon] \in \Hr$ is strictly positive, because 
$$
\{n \in \N: \tfrac{1}{n} > 0  \} = \N \in \mc{F}.
$$
If $r$ is any positive real number, then the set 
$$
\{n\in \N: \tfrac{1}{n}< r\} \in  \mc{F}
$$ 
is cofinite, since $\mc{F}$ is a nonprincipal ultrafilter. 
Therefore $\epsilon$ is a \emph{positive infinitesimal}, that is, a positive element that is smaller than all positive real numbers, 
$$
[(0,0,\ldots)]<[\epsilon] < [(r,r,\ldots)].  
$$

On the other,  given a diverging sequence
$$\omega = (1,2,3,\ldots), 
$$ 
 $[\omega]$ is  a \emph{positive infinite element}. 
More generally, all equivalence classes of sequences that converge to $0$ are infinitesimal elements in $\Hr$, 
and all equivalence classes of diverging sequences are infinite elements.

\begin{example}[Ultrafilters on $\N$]\label{ex:ultrafilter}\quad
\begin{enumerate}[label=(\roman*),ref=(\roman*)]
\item By fixing one natural number, say 5, one can define an ultrafilter as all subsets of the naturals that contain the element 5. This ultrafilter is \emph{principal} since it contains finite subsets (the subset $\{5\}$ and $\{5,23\}$ for example). 
If one would define $\Hr$ by defining equivalence classes of $\R^{\N}$ using this (principal) ultrafilter, the result will be exactly the usual reals $\R$, since any two elements $a:=(a_n),b:=(b_n)$ are identified whenever $a_5=b_5$, so the rest of the sequence can be ignored. This illustrates why it is crucial that the ultrafilter is nonprincipal.

\item It is not possible to write down an example of a nonprincipal ultrafilter, as it would require to make use of the axiom of choice. On top of all cofinite subsets, one has to choose for all possible ways to divide the natural numbers into two infinite sets which one is in the ultrafilter, in a way that is consistent with the definition (regarding supersets and intersections). 
For example, one has to choose which of the two alternating sequences in (A.1) is identified with $(0,0,0,\ldots)$ and which one with $(1,1,1\ldots)$. 
\end{enumerate}
\demo\end{example}

%%====================================================
\section{Tensor stable positivity on $\ell^2$}
\label{app:l2}

In this appendix we prove the existence of essential tsp maps on the sequence space $\ell^2_\C$ (\cref{thm:l2}). 
$\ell^2_\C$ is the subspace of $\C^{\N}$ of all sequences $(x_n)$ with  $\sum_{n=1}^{\infty} |x_n|^2 < \infty$, i.e.\ which are square summable.
We denote the corresponding subset with elements from $\R$ by $\ell^2$.

We start by  defining a notion of positivity in $\ell^2$ (\cref{sec:l21}), then 
positive semidefinite matrices over $\ell^2_\C$ (\cref{sec:l22}),  
positive linear maps and tsp maps on $\ell^2_\C$ (\cref{sec:l24}), 
 and finally prove the existence of essential tsp maps (\cref{app:tsp_l2}). 
We also explore the existence of an inner product on $\ell^2_\C$ (\cref{app:l2_inner}). 

%%=====================================
\subsection{Positivity on $\ell^2$}
\label{sec:l21}

In order to define positivity on $\ell^2$, we fix a nonprincipal ultrafilter $\F$ on $\N$ (\cref{def:ultrafilter}) and use it for all upcoming definitions. 

\begin{definition}[Positivity of elements in $\ell^2$]\label{def:l2pos}
Given an ultrafilter $\F$, an element $(x_n)  \in \ell^2$ is called \emph{$\ell^2$-nonnegative}, denoted  
$(x_n) \geq_{\ell^2} 0$,  if 
$$
\{n: x_n \geq 0\} \in \mc{F}.
$$
\end{definition}

Note that with this definition, 
$\ell^2$-nonpositive ($\leq_{\ell^2}$), 
$\ell^2$-negative ($<_{\ell^2}$), 
$\ell^2$-positive ($>_{\ell^2}$)  are also defined---for example, 
$$
a <_{\ell^2} b \quad \textrm{if} \quad \{n : a_n < b_n \} \in \mc{F}.
$$ 
Note also that it can happen that  $a \leq_{\ell^2} b$ and $b \leq_{\ell^2} a$ even though $a \neq b$, namely when they are equal on an index set that is in the ultrafilter. 

This definition of positivity gives rise to a total order: 

\begin{lemma}[Ring ordering]\label{lem:ring}
The subset of $\ell^2$-nonnegative elements, 
$$
T: = \{x \in \ell^2: x \geq_{\ell^2} 0\}, 
$$ 
defines a total ring-ordering on $\ell^2$.
\end{lemma}

\begin{proof}
A total ordering $T\subset R$ for a ring $R$ without $1$ has the following properties: 
\begin{enumerate}
\item\label{ord1} $T + T \subseteq T$ (a sum of nonnegative elements is positive),
\item\label{ord2} $T \cdot T \subseteq T$ (a product of nonnegative elements is nonnegative),
\item\label{ord3} $R^2 \subseteq T$ (all squares are nonnegative), 
\item\label{ord4} $T \cup -T = R$ (the union of nonnegative and nonpositive elements form the total ring).
\item\label{ord5} $T \cap -T$ is a prime ideal. 
\end{enumerate}
It is easy to check that properties (\ref{ord1}) - (\ref{ord4}) are fulfilled. The ordering has a nontrivial support $T \cap -T$ that consists of all elements $(x_n) \in \ell^2$ for which $\{n: x_n=0\} \in \F$. To see that this set is a prime ideal, consider two elements 
$(a_n),(b_n) \in \ell^2$ such that $(a_n)\cdot(b_n) \in T \cap -T$, meaning that 
$$
C:=\{n: a_nb_n=0\}\in \F.
$$ 
Assume now, towards a contradiction, that both $A:= \{n:a_n=0\}\notin\F$ and $B:=\{n:b_n=0\}\notin\F$. 
Since  $\F$ is an ultrafilter, we have that $\N\setminus A\in \F$ and $\N\setminus B\in \F$, 
so 
$$
D:= (\N\setminus A)\cap (\N\setminus B) \in \F,
$$ 
but then $\N\setminus D = A \cup B =  C \notin \F$, which is a contradiction. So either $A$ or $B$ are in $\F$, showing that $\F$ is a prime ideal.
\end{proof}

From now on we consider the complex sequences $\ell^2_\C$. In fact we consider a tuple over $\ell^2_\C$, that is, an element of  $(\ell^2_\C)^d$.

\begin{definition}[Quasi-inner product]\label{def:l2inner}
The \emph{quasi-inner product}, denoted $\langle \cdot , \cdot \rangle_{\textrm{seq}}$, 
on $(\ell^2_\C)^d$ is a map 
$$\langle \cdot , \cdot  \rangle_{\textrm{seq}}: (\ell^2_\C)^d \times (\ell^2_\C)^d \rightarrow \ell^2 
$$ 
where 
$$
\langle a,b \rangle_{\textrm{seq}} = (\langle a_n,b_n\rangle), 
$$
where the right hand side uses the standard inner product on $\C^d$ entrywise.
\end{definition}

Since the image of the quasi-inner product is not a field, this is not an inner product. It does however satisfy the other properties of an inner product: it is linear (even $\ell^2$-linear) in the first argument, conjugate symmetric and 
positive definite, namely for $\ell^2_\C \ni a \neq 0$   
$$
\langle a,a \rangle_{\textrm{seq}}>_{\ell^2} 0.
$$ 
Note that this differs from the standard inner product on $\ell^2$ as a Hilbert space (\cref{def:l2standard}). 

%%========================================
\subsection{Matrices on $\ell^2$}
\label{sec:l22}
We now define a notion of psd matrices over $\ell^2_\C$. 

\begin{definition}[Psd over $\ell^2_\C$]
A matrix $A \in \mc{M}_d(\ell^2_\C)$ is called \emph{$\ell^2$-psd}, denoted $A \geqslant_{\ell^2} 0$, if 
$$
\langle v, Av \rangle_{\textrm{seq}} \geq_{\ell^2} 0
$$ for all $v \in (\ell^2_\C)^d$. 
\end{definition}

Note that the symbol for $\ell^2$-psd is $\geqslant_{\ell^2}$ whereas the symbol for $\ell^2$-nonnegative is $\geq_{\ell^2}$, in analogy to psd matrices ($\geqslant$) and nonnegative numbers ($\geq$).
 
A matrix $A \in \mc{M}_d(\ell^2_\C)$ can be seen as a matrix with elements from $\ell^2_\C$, or as a sequence of matrices $A_n$ with elements in $\C$ (obeying the square-summability condition; \cref{fig:l2}). 
One can think of each  $A_n$ as a `layer' of $A$. 

\begin{figure}[t]\centering
\includegraphics[width = 0.7\textwidth]{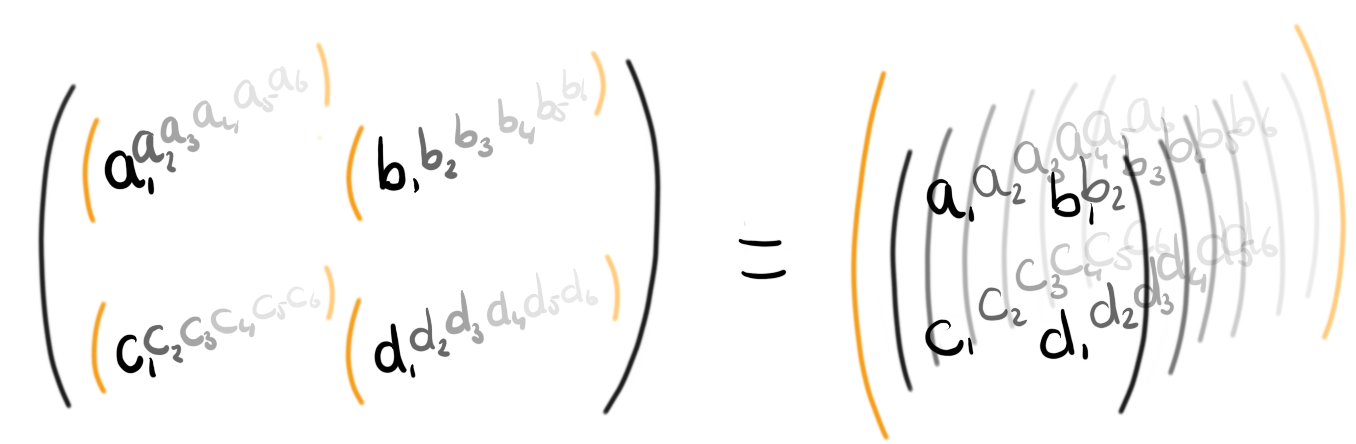}
\caption{Matrices over $\ell^2$ are equivalent to sequences of matrices, with a condition of square-summability on the matrix elements.}
\label{fig:l2}
\end{figure}

\begin{lemma}[Psd of layers] \label{lem:matrix_layers} 
Given a matrix $A \coloneqq (A_n) \in \Ml$, the following are equivalent:
\begin{enumerate}[label=(\roman*)]
\item \label{lem:matrix_layers:i} 
$A \geqslant_{\ell^2} 0$. 
\item  \label{lem:matrix_layers:ii} 
$\{n: A_n \geqslant 0\} \in \F$. 
\end{enumerate}
\end{lemma}

\begin{proof}
We prove the contrapositive of \ref{lem:matrix_layers:i} $\Rightarrow$ \ref{lem:matrix_layers:ii}.  
Consider a matrix $A: = (A_n) \in \Ml$ such that $\{n: A_n \geqslant 0\} \notin \F$, and define the complement of this set as 
$$
X\coloneqq \{n: A_n \ngeqslant 0\} \in \F. 
$$
Then for all $m\in X$ there is a $w_m\in \C^d$ such that 
$$
\langle w_m, A_m w_m\rangle <0.
$$ 
We now construct $v = (v_n)$ by setting $v_n \coloneqq w_n$ for all $n\in X$ and fill the rest of the entries arbitrarily. 
To ensure that $v \in  (\ell^2_\C)^d$, we rescale $w_n$ with a factor depending on $n$, such that the $\ell^2$ condition of square summability is satisfied, resulting 
a vector $v \in  (\ell^2_\C)^d$ for which 
$$
\langle v, Av \rangle <_{\ell^2} 0.
$$   
It follows that $A$ is not $\ell^2$-psd, namely $A \ngeqslant_{\ell^2}0$. 

\ref{lem:matrix_layers:ii} $\Rightarrow$ \ref{lem:matrix_layers:i}. Given a matrix $A := (A_n)  \in \Ml$, define the set of entries whose layer is psd as 
$$
Y\coloneqq \{n: A_n \geqslant 0\} \in \F.
$$
 For every $w= (w_n) \in  (\ell^2_\C)^d$, define 
 $$ Z_w\coloneqq \{n:  \langle w_n,A_n w_n \rangle \geq 0\} 
 $$ 
 We know that for all such $w$,  $Z_w\supseteq Y$, 
 since $Z_w$ considers a contraction with a specific element $w$ whereas $Y$ considers a contraction with all elements, so $Z_w$ contains the indices of  $Y$ and perhaps more.   
 Since the ultrafilter is closed under supersets, it follows that $ Z_w\in \F$ for all such $w$. 
 Therefore, $\langle w, Aw \rangle \geq_{\ell^2} 0$ for all $w \in  (\ell^2_\C)^d$, that is, $A\geqslant_{\ell^2} 0$. 
\end{proof}

%%======================================
\subsection{Tensor stable positivity on $\ell^2_\C$}
\label{sec:l24}

We now consider linear maps 
$$
\Pp: \Mm_{d}(\ell^2_\C) \rightarrow \Mm_{d}(\ell^2_\C).
$$ 
As with matrices and vectors, we are interested in linear maps $\Pp$   that act `layerwise', i.e. 
\be\label{eq:layerlinear}
\Pp\coloneqq (\Pp_n) ~\textrm{where} ~\Pp_n: \mc{M}_d(\C) \rightarrow \mc{M}_d(\C).
\ee
Not every map linear map acts layerwise, but exactly the $\ell^2_\C$-linear maps (i.e.\ linear under multiplication with elements from $\ell^2_\C$) have this property. 

The image of an  $\ell^2_\C$-linear map $\Pp$ is in $\Mm_{d}(\ell^2_\C)$ if  $\Pp$ is \emph{uniformly bounded}, meaning that there exists a common bound on $||\Pp_n||_{\textrm{op}}$ for all $n$. Here the operator norm of a linear map $A: V \rightarrow W$ is defined as usual
$$
||A||_{\textrm{op}} = \textrm{inf}\{c\geq 0: |Av| \leq c|v| \:\: \forall v\in V\}.
$$ 
The uniformly bounded maps satisfying \eqref{eq:layerlinear} are called \emph{uniformly bounded linear} maps. 

Let us now define $\ell^2$-positivity of uniformly bounded linear maps.  

\begin{definition}[Positivity of uniformly bounded linear maps]
A uniformly bounded linear map $\Pp$ is \emph{positive}, denoted  $\Pp \succcurlyeq_{\ell^2} 0$, if it maps $\ell^2$-psd matrices to $\ell^2$-psd matrices, that is, 
$$
A \geqslant_{\ell_2} 0 \implies \Pp(A) \geqslant_{\ell_2} 0.
$$
\end{definition}

Note that we again follow the convention of denoting the positivity of maps by $\succcurlyeq_{\ell^2}$, in analogy with positive maps ($\succcurlyeq$).

\begin{lemma}[Positive maps under layers]
Given a uniformly bounded linear map $\Pp=(\Pp_n)$ 
the following two statements are equivalent: 
\begin{enumerate}[label=(\roman*)]
\item  $\Pp \succcurlyeq_{\ell^2} 0$.
\item $\{n: \Pp_n \succcurlyeq  0\} \in \F$. 
\end{enumerate}
\end{lemma}

The proof is analogous to that of \cref{lem:matrix_layers}.

We now define a tensor product on $\ell^2$ in the expected way.  

\begin{definition}[Tensor product on $\ell^2$]
The \emph{tensor product} on $(\ell^2_\C)^d$, denoted $\otimes_{\ell^2}$, is the bilinear map  
\begin{align*} 
  \otimes_{\ell^2}  : (\ell^2_\C)^d \times (\ell^2_\C)^d & \rightarrow (\ell^2_\C)^{d^2} 
  \end{align*}
  where
\begin{align*}  
 (a_n)    \otimes_{\ell^2} (b_n) \coloneqq (a_n \otimes b_n)  
 \end{align*}
where the right hand side uses the standard tensor product on $\C^d$ entrywise. 
\end{definition}

In order to define a notion of $\ell^2$-positivity of linear maps on $\mc{M}_d(\ell^2_\C) \otimes_{\ell^2}  \mc{M}_d(\ell^2_\C)$ we use the natural isomorphism 
$$
\mc{M}_d(\ell^2_\C) \otimes_{\ell^2} \mc{M}_d(\ell^2_\C) \cong \mc{M}_{d^2}(\ell^2_\C). 
$$
This equivalence allows to use the $\ell^2$-positivity from the right hand side on the tensor products on the left hand side. Namely, a linear map 
is $\ell^2$-positive if it maps the set of $\ell^2$-psd matrices in $\mc{M}_d(\ell^2_\C) \otimes_{\ell^2} \mc{M}_d(\ell^2_\C)$ to itself. 
With this can define $\ell^2$-tsp as one would expect: 

\begin{definition}[$\ell^2$-tsp] 
Let $\Pp: \mc{M}_d(\ell^2_\C) \rightarrow \mc{M}_d(\ell^2_\C)$ be a uniformly bounded linear map. 
\begin{enumerate}[label=(\roman*),ref=(\roman*)]
\item 
$\Pp$ is \emph{$\ell^2$-$n$-tsp}  if $\Pp^{\otimes_{\ell^2}  n}$  
is $\ell^2$-positive.  
\item 
$\Pp$ is \emph{$\ell^2$-tsp} if $\Pp^{\otimes_{\ell^2} n}$ 
is $\ell^2$-positive for all $n\in\N$.  
\end{enumerate}
\end{definition}

For the following we denote the identity map on $\Mm_{d}(\ell^2_\C)$  
 by  $\textrm{id}_d$, and the transposition map by 
\begin{align*}
\theta_d:  \Mm_{d}(\ell^2_\C) &\rightarrow \Mm_{d}(\ell^2_\C)\\
(A_n) & \mapsto (A_n^T), 
\end{align*}
where  
${}^T$  denotes  the usual transposition on $\mc{M}_d(\C)$. 

By Choi's Theorem, complete positivity of a map $\Pp: \mc{M}_d(\C) \rightarrow \mc{M}_d(\C)$ is equivalent to $d$-positivity of the map (\cref{ssec:tsp}). 
We use this result to define $\ell^2$-complete positivity. 

\begin{definition}[$\ell^2$-completely (co)positive map]\label{def:l2cp} 
Let $\Pp: \Mm_{d}(\ell^2_\C) \rightarrow \Mm_{d}(\ell^2_\C)$ be a  uniformly bounded linear map. 
\begin{enumerate}[label=(\roman*), ref=(\roman*)]
\item 
$\Pp$ is \emph{$\ell^2$-completely positive} if 
$$
\Pp \otimes_{\ell^2}  \textrm{id}_d \succcurlyeq_{\ell^2} 0.
$$ 
\item $\Pp$  is \emph{$\ell^2$-completely co-positive} if 
$\Pp=\theta_d \circ \mc{S}$ for some   $\ell^2$-completely positive map $\mc{S}$.  
\end{enumerate}
\end{definition}

\begin{lemma}[$\ell^2$-completely (co-)positivity under layers]\label{lem:trivial_layers}
Let $\Pp: \Mm_{d}(\ell^2_\C) \rightarrow \Mm_{d}(\ell^2_\C)$ be a  uniformly bounded linear map. 
$\Pp$ is $\ell^2$-completely (co-)positive if and only if 
$$
\{n: \Pp_n ~\textrm{is completely (co-)positive} \} \in \F.
$$ 
\end{lemma}

\begin{proof}
This follows from the behavior of $\ell^2$-tensor products and $\ell^2$-positivity under the layers.
\end{proof}

We call $\ell^2$-completely positive and $\ell^2$-completely co-positive maps  \emph{trivial} $\ell^2$-tsp maps, and those which are not trivial \emph{essential}  $\ell^2$-tsp maps. 

\begin{lemma}[Trivial $\ell^2$-tsp maps]\label{lem:trivial_maps}
$\ell^2$-completely positive and $\ell^2$-completely co-positive maps are $\ell^2$-tsp.
\end{lemma}

\begin{proof}
By the behavior of $\ell^2$-tensor products and $\ell^2$-positivity, a map is $\ell^2$-tsp if 
$$
\{n: \Pp_n ~\textrm{is tsp} \} \in \F. 
$$
Using \cref{lem:trivial_layers} and the fact that completely (co-)positive maps from $\mc{M}_d(\C)$ to $\mc{M}_d(\C)$ are tsp, this concludes the proof. 
\end{proof}

%%%=====================
\subsection{Existence of essential tsp on $\ell^2$}
\label{app:tsp_l2}

Based on all definitions and results above, we now show that essential tsp maps exist over $\ell^2_\C$. 

\begin{theorem}[Essential $\ell^2$-tsp maps] \label{thm:l2}
There exist essential $\ell^2$-tsp uniformly bounded linear maps $\Pp: \Mm_{d}(\ell^2_\C) \rightarrow \Mm_{d}(\ell^2_\C)$. 
\end{theorem}

\begin{proof}
For every $n$ there exist an essential $n$-tsp map $\Pp_n: \Mm_d(\C) \rightarrow \Mm_d(\C)$ \cite{Mu16}. 
For every $n$ we fix such a map, and construct the uniformly bounded linear map 
$$
\Pp = (\Pp_n): \Mm_{d}(\ell^2_\C) \rightarrow \Mm_{d}(\ell^2_\C). 
$$
We may need to rescale every $\Pp_n$ by a constant  factor to enforce the uniform boundedness condition. 
By \cref{lem:trivial_layers} the map $\Pp$ is essential. 
By the definition of the $\ell^2$-tensor product the $m$-th tensor power of this map is 
$$
\Pp^{\otimes_{\ell^2}  m} = (\Pp_n^{\otimes m}). 
$$ 
Moreover, for any $m$, $\Pp_n^{\otimes m} \succcurlyeq 0$ for all   $n\geq m$.   
Therefore 
$$
\{n: \Pp_n^{\otimes m} \succcurlyeq 0 \} \in \F \quad \forall m, 
$$
since this is a co-finite subset of $\N$. By the definition of $\ell^2$-positivity of linear maps, we conclude that $\Pp^{\otimes_{\ell^2}  m}$ is positive for all $m$, and is therefore essential $\ell^2$-tsp. 
\end{proof}

Note that the construction heavily relies on our chosen notion of positivity, since we specifically use the ultrafilter. 

%%======================================
\subsection{Inner product on $\ell^2$}
\label{app:l2_inner}

Following  \cref{sec:hyper}, we  explore the existence of NPT bound entangled states on $\ell^2$, using the fact that this space is an infinite dimensional Hilbert space. 
The standard inner product on $\ell^2$ is defined as follows: 

\begin{definition}[Standard  inner product] \label{def:l2standard} 
Given $a,b \in \ell^2_\C$,  the \emph{standard inner product}, denoted $\langle \: , \:\rangle_{\textrm{st}}$, is given by 
$$
\langle a,b\rangle_{\textrm{st}} = \sum_{n=1}^{\infty} \bar{a}_n b_n.
$$
\end{definition}

It is immediate to show  that this inner product comes with a notion of positivity that does not coincide with the positivity of the quasi-inner product $\langle \cdot,\cdot\rangle_{\textrm{seq}}$ of \cref{def:l2inner}. 
For example, for  the following elements in $\ell^2 $
$$
u = (-1,0,0,0\ldots), \quad v = (1,0,0,0,\ldots)
$$
we have that $\langle u,v\rangle_{\textrm{st}} = -1 $,  
while 
$$
\langle u,v\rangle_{\textrm{seq}} = (-1,0,0,0,\ldots) \geq_{\ell^2} 0.
$$
This is not only the case for the standard inner product, but any inner product on this space will have the same problem:

\begin{proposition}[Impossibility of inner product]\label{pro:inner} 
There does not exist an inner product $\langle\cdot,\cdot\rangle_{\ell^2}: \ell^2 \times \ell^2 \rightarrow \R$  
such that 
$$
\langle x,y\rangle_{\ell^2} \geq 0 \quad \textrm{iff} \quad  \langle x,y\rangle_{\textrm{seq}} \geq_{\ell^2} 0 
$$ 
for all $x,y \in \ell^2$. 
\end{proposition}

\begin{proof}
An inner product must be linear in every component, conjugate symmetric and positive definite. 
We claim that for every inner product $\langle\cdot,\cdot\rangle_{\ell^2}$ there are $x,y \in \ell^2$ such that either 
$\langle x,y\rangle_{\ell^2} $ is positive and  $ \langle x,y\rangle_{\textrm{seq}}$ is not, or the other way around. 
Consider the following elements in $\ell^2$:
$$
x = (1,0,0,0,\ldots) , \quad y = (0,y_1, y_2, y_3, \ldots)  
$$  
 with all $y_i\in \R$. 
 The quasi-inner product between $x+\eps y$ and $x-\eps y$ for some $\eps\in \R$ yields: 
\begin{eqnarray}
\nonumber \langle x+ \eps y, x - \eps y \rangle_{\textrm{seq}} &=& (x+\eps y)\cdot(x-\eps y) \\\nonumber&=& x^2 - \eps^2 y^2 \\
&=& (1, -\eps^2 y_1^2, -\eps^2 y_2^2,-\eps^2 y_3^2, \ldots) <_{\ell^2} 0. 
\label{eq:innerprodseq}
\end{eqnarray}  
However, by linearity,  
\be \label{eq:innerprodl2}
\langle x+ \eps y, x - \eps y \rangle_{\ell^2} = \langle x,x\rangle_{\ell^2} - \eps^2 \langle y,y\rangle_{\ell^2},  
\ee
and by positive definiteness, both $\langle x,x\rangle_{\ell^2}>0$ and $\langle y,y\rangle_{\ell^2}>0$. 
For small enough  $\eps$, Equation \eqref{eq:innerprodl2} is positive whereas Equation \eqref{eq:innerprodseq} is negative.   
\end{proof} 

Since a suitable inner product fails to exist already for single elements in $\ell^2$, there will not exist a suitable matrix-inner product either. 
One can therefore not interpret terms like $\textrm{tr}(\rho A)$ as the probability of obtaining an outcome of a quantum measurement for an observable $A \in \mc{M}(\ell^2_\C)$ and a quantum state $\rho \in \mc{M}(\ell^2_\C)$.

\end{appendices}

%%==============================================


\begin{thebibliography}{10}

\bibitem{Be19}
V.~Benci, L.~L.~Baglini, and K.~Simonov.
\newblock {Infinitesimal and infinite numbers as an approach to quantum
  mechanics}.
\newblock {\em Quantum}, 3:137, 2019.


\bibitem{Ch06}
I.~Chattopadhyay and D.~Sarkar.
\newblock {NPT Bound Entanglement--The Problem Revisited}.
\newblock {\em \href{https://arxiv.org/abs/quant-ph/0609050}{arXiv:quant-ph/0609050}}, 2006.

\bibitem{Ch19b}
G.\ Chiribella and R.\ W.\ Spekkens (Eds.). 
\newblock {Quantum Theory: Informational Foundations and Foils},
\newblock Springer, 2019.

\bibitem{Ch75}
M.~Choi.
\newblock {Positive Linear Maps on Complex}.
\newblock {\em Linear Algebra Appl}, 10:285, %10:285--290, 
1975.

\bibitem{Ch18}
M.~Christandl, A.~Lucia, P.~Vrana, and A.~H.~Werner.
\newblock {Tensor network representations from the geometry of entangled
  states}.
\newblock {\em SciPost Phys.}, 9:042, 2018.

\bibitem{Cu15b}
T.~S.~Cubitt, D.~Perez-Garcia, and M.~M.~Wolf.
\newblock {Undecidability of the spectral gap}.
\newblock {\em Nature}, 528:207, %528:207--211, 
2015.

\bibitem{De15}
G.~{De las Cuevas}, T.~S. Cubitt, J.~I. Cirac, M.~M.\ Wolf, and
  D.~P{\'{e}}rez-Garc{\'{i}}a.
\newblock {Fundamental limitations in the purifications of tensor networks}.
\newblock {\em J. Math. Phys.}, 57:071902, 2016.

\bibitem{De20d}
G.~{De las Cuevas}.  
\newblock {Universality everywhere implies undecidability everywhere},
\newblock {\em \href{https://fqxi.org/community/forum/topic/3529}{FQXi Essay}},
2020.

\bibitem{De20}
G.~{De las Cuevas}, T.~Fritz, and T.~Netzer.
\newblock {Optimal Bounds on the Positivity of a Matrix from a Few Moments}.
\newblock {\em Comm. Math. Phys.}, 375:105,  % 375:105--126, 
2020.

\bibitem{De19d}
G.~{De las Cuevas}, M.~{Hoogsteder Riera}, and T.~Netzer.
\newblock {Tensor decompositions on simplicial complexes with invariance}.
\newblock {\em \href{https://arxiv.org/abs/1909.01737}{arXiv:1909.01737}}, 2019.

\bibitem{De19e}
F.\ Del Santo and N.\ Gisin. 
\newblock { Physics without determinism: Alternative interpretations of classical physics}. 
\newblock {\em Phys. Rev. A}, 100:062107, 2019. 
% https://doi.org/10.1103/PhysRevA.100.062107

\bibitem{Di00}
D.~P.\ DiVincenzo, P.\ W.\ Shor, J.\ A.\ Smolin, B.~M.\ Terhal, and A.~V.\ Thapliyal.
\newblock {Evidence for bound entangled states with negative partial
  transpose}.
\newblock {\em Phys. Rev. A}, 61:062312, 2000.

\bibitem{Du00}
W.~D{\"{u}}r, J.~I.\ Cirac, M.~Lewenstein, and D.~Bru{\ss}.
\newblock {Distillability and partial transposition in bipartite systems}.
\newblock {\em Phys. Rev. A}, 61:062313, 2000.

\bibitem{Ei11}
J.~Eisert, M.~P.~M{\"{u}}ller, and C.~Gogolin.
\newblock {Quantum measurement occurrence is undecidable}.
\newblock {\em Phys. Rev. Lett.}, 108:260501, 2012.

\bibitem{Fi16}
S.N. Filippov, and K. Y. Magadov
\newblock {Positive tensor products of qubit maps and n-tensor-stable positive qubit maps} 
\newblock{\em J. Phys. A: Math. Theor.} {\bf 50} 055301,
2016.

\bibitem{Gh08}
S.~Gharibian.
\newblock {Strong NP-Hardness of the Quantum Separability Problem}.
\newblock {\em Quantum Inf. Comput.}, 10:0343, %10:0343--0360, 
2008.

\bibitem{Gi18}
N.\ Gisin. 
\newblock {Indeterminism in Physics, Classical Chaos and Bohmian mechanics: Are real numbers really real?} 
\newblock {\em Erkenn.}, 2019.


\bibitem{Gi19}
N.\ Gisin. 
\newblock {Real numbers are the hidden variables of classical mechanics}. 
\newblock {\em Quantum Stud.: Math. Found.}, 7, 2019.


\bibitem{Go98}
R.~Goldblatt.
\newblock {Lectures on the hyperreals},
\newblock Springer, 1998.

\bibitem{Gu03}
L.~Gurvits.
\newblock {Classical deterministic complexity of Edmonds' problem and quantum
  entanglement}.
\newblock {\em Proc. Annu. ACM Symp. Theory Comput.}, 10, %10--19, 
2003.

\bibitem{Ho98}
M.~Horodecki, P.~Horodecki, and R.~Horodecki.
\newblock {Mixed-state entanglement and distillation: Is there a “Bound”
  entanglement in nature?}
\newblock {\em Phys. Rev. Lett.}, 80:5239, %80:5239--5242, 
1998.

\bibitem{Ho20}
P.~Horodecki, {\L}.~Rudnicki, and K.~{\.{Z}}yczkowski.
\newblock {Five open problems in quantum information}.
\newblock {\em \href{https://arxiv.org/abs/2002.03233}{arXiv:2002.03233}}, 2020.

\bibitem{Io07}
L.~M.\ Ioannou.
\newblock {Computational complexity of the quantum separability problem}.
\newblock {\em Quantum Inf. Comput.}, 7:335, %7:335--370, 
2007.

\bibitem{Kl14}
M.~Kliesch, D.~Gross, and J.~Eisert.
\newblock {Matrix-product operators and states: NP-hardness and
  undecidability}.
\newblock {\em Phys. Rev. Lett.}, 113:160503, 2014.

\bibitem{Ko97}
D.~C.~Kozen,  
\newblock {\href{https://doi.org/10.1007/978-1-4612-1844-9}{Automata and Computability}}. 
\newblock {Springer}, 1997. 

\bibitem{Mu18}
A.~M{\"{u}}ller-Hermes.
\newblock {Decomposability of linear maps under tensor powers}.
\newblock {\em J. Math. Phys.}, 59:102203, 2018.

\bibitem{Mu16}
A.~M{\"{u}}ller-Hermes, D.~Reeb, and M.~M.~Wolf.
\newblock {Positivity of linear maps under tensor powers}.
\newblock {\em J. Math. Phys.}, 57:015202, 2016.

\bibitem{Pa10}
{\L}.~Pankowski, M.~Piani, M.~Horodecki, and P.~Horodecki.
\newblock {A few steps more towards NPT bound entanglement}.
\newblock {\em IEEE Trans. Inf. Theory}, 56:4085, %56:4085--4100, 
2010.

\bibitem{Pr01}
A.~Prestel and C.~N.~Delzell.
\newblock {{Positive Polynomials}}.
\newblock Springer, 2001.

\bibitem{Re21}
M.-O.\ Renou, D.\ Trillo, M.\ Weilenmann, L.\ P.\ Thinh, A.\ Tavakoli, N.\ Gisin, A.\ Acin, and M.\ Navascues. 
\newblock {Quantum physics needs complex numbers}. 
\newblock {\em \href{https://arxiv.org/abs/2101.10873}{arXiv:2101.10873}}, 2021. 

\bibitem{Sc21}
M.~Scandi and J.~Surace.
\newblock {Undecidability in resource theory: can you tell theories apart?}
\newblock {\em \href{https://arxiv.org/abs/2105.09341}{arXiv:2105.09341}}, 2021.

\bibitem{Sc81}
A.~Sch{\"{o}}nhage.
\newblock {Partial and Total Matrix Multiplication}.
\newblock {\em SIAM J. Comput.}, 10:434, %10:434--455, 
1981.

\bibitem{St18}
E.~St{\o}rmer.
\newblock {Mapping cones and separable states}.
\newblock {\em Positivity}, 22:493, %22:493--499, 
2018.

\bibitem{St86}
V.~Strassen.
\newblock {Asymptotic spectrum of tensors and the exponent of matrix
  multiplication}.
\newblock In {\em Proc. Annu. IEEE Symp.}, 49, %49--54, 
1986.

\bibitem{Wo11}
M.~M.~Wolf, T.~S.~Cubitt, and D.~Perez-Garcia.
\newblock {Are problems in Quantum Information Theory (un)decidable?}
\newblock {\em \href{https://arxiv.org/abs/1111.5425}{arXiv:1111.5425}}, 2011.



\end{thebibliography}
\end{document}